\pdfoutput=1
\documentclass[format=acmsmall, review=false]{acmart}
\usepackage{acm-ec-20}
\usepackage{booktabs} 
\usepackage[ruled]{algorithm2e} 
\usepackage{algorithmic}

\SetAlFnt{\small}
\SetAlCapFnt{\small}
\SetAlCapNameFnt{\small}
\SetAlCapHSkip{0pt}
\IncMargin{-\parindent}

\usepackage{tikz}
\usepackage{enumitem}
\usepackage{subcaption}
\usepackage{pgfplots}

\usepackage{cleveref}
\usepackage{pdflscape}
\usepackage{xspace}

\usepackage{bbm}

\setcitestyle{authoryear}

\allowdisplaybreaks

\newtheorem{theorem}{Theorem}[section]
\newtheorem*{theorem*}{Theorem}
\newtheorem{corollary}[theorem]{Corollary}

\newtheorem{lemma}[theorem]{Lemma}
\newtheorem*{lemma*}{Lemma}

\theoremstyle{definition}
\newtheorem{definition}[theorem]{Definition}
\newtheorem{example}[theorem]{Example}
\newtheorem{remark}[theorem]{Remark}

\DeclareMathOperator*{\argmax}{arg\,max}
\DeclareMathOperator*{\leximin}{leximin}
\newcommand{\wefxy}{WEF$(x,y)$\xspace}
\newcommand{\wpropxy}{WPROP$(x,y)$\xspace}
\newcommand{\wpropstar}[1]{WPROP$^*{#1}$}
\newcommand{\wpropstarxy}{WPROP$^*(x,y)$\xspace}
\newcommand{\nnmms}{$1/n$-NMMS\xspace}

\newcommand{\floor}[1]{\left\lfloor #1 \right\rfloor}
\newcommand{\ceil}[1]{\left\lceil #1 \right\rceil}

\title[Weighted Fairness Notions for Indivisible Items Revisited]{Weighted Fairness Notions for Indivisible Items Revisited}

\author{Mithun Chakraborty}
\affiliation{%
  \institution{University of Michigan, USA}}

\author{Erel Segal-Halevi}
\affiliation{%
  \institution{Ariel University, Israel}}

\author{Warut Suksompong}
\affiliation{%
  \institution{National University of Singapore, Singapore}}

\begin{abstract}
We revisit the setting of fairly allocating indivisible items when agents have different weights representing their entitlements.
First, we propose a parameterized family of relaxations for weighted envy-freeness and the same for weighted proportionality; the parameters indicate whether smaller-weight or larger-weight agents are given a higher priority. 
We show that each notion in these families can always be satisfied, but any two cannot necessarily be fulfilled simultaneously.
We then introduce an intuitive weighted generalization of maximin share fairness and establish the optimal approximation of it that can be guaranteed.
Furthermore, we characterize the implication relations between the various weighted fairness notions introduced in this and prior work, and relate them to the lower and upper quota axioms from apportionment.
\end{abstract}


\begin{document}

\maketitle

\section{Introduction}

Research in fair division is quickly moving from the realm of theory to practical applications, such as division of various assets between individuals \citep{goldman2015spliddit}, distribution of food to charities \citep{aleksandrov2015online} and medical equipment among communities \citep{pathak2021fair}, 
and allocation of course seats among students \citep{othman2010finding}.
Two complicating factors in such applications are that items may be \emph{indivisible}, and recipients may have \emph{different entitlements}.
For example, when dividing food packs or medical supplies among organizations or districts, it is reasonable to give larger shares to recipients that represent more individuals.\footnote{
An interesting historical case of division with different entitlements occurred during the dissolution of Czechoslovakia
in 1993, where the federal assets were divided between the Czech Republic and Slovakia in a ratio of 2:1 reflecting their populations
\citep{Engelberg93}.
}
This raises the need to define appropriate fairness notions taking these factors into account.

With divisible resources and equal entitlements,
two well-established fairness benchmarks are \emph{envy-freeness}---no agent prefers the bundle of another agent, 
and \emph{proportionality}---every agent receives at least $1/n$ of her value for the set of all resources, where $n$ denotes the number of agents.
When the resources consist of indivisible items and all of them must be allocated, neither of these benchmarks can always be met, for example if one item is extremely valuable in the eyes of all agents.
This has motivated various relaxations, of which we next describe two popular approaches.
The first approach allows hypothetically removing or adding a single item before applying the fairness notion. 
In particular,
\emph{envy-freeness up to one item (EF1)} requires that if an agent envies another agent, the envy should disappear upon removing one item from the envied agent's bundle \citep{LiptonMaMo04,Budish11}.
\emph{Proportionality up to one item (PROP1)} demands that if an agent falls short of the proportionality benchmark, this should be rectified after adding one item to her bundle \citep{ConitzerFrSh17,AzizMoSa20,AzizCaIg22}. 
The second approach modifies the fairness threshold itself: instead of $1/n$ of the total value, the threshold becomes the \emph{maximin share (MMS)}, which is the maximum value that an agent can ensure herself by partitioning the items into $n$ parts and receiving the worst part \citep{Budish11}. 
The existence guarantees with respect to these relaxations are quite well-understood by now---an allocation satisfying EF1 and PROP1 always exists \citep{LiptonMaMo04,CaragiannisKuMo19}, and even though the same is not true for MMS fairness, there is always an allocation that gives every agent a constant fraction of her MMS \citep{KurokawaPrWa18,GargTa21,GhodsiHaSe21}.

Recently, there have been several attempts to extend these approximate notions to agents with different entitlements.
However, the resulting notions have been shown to exhibit a number of unintuitive and perhaps unsatisfactory features.
In the ``fairness up to one item'' approach, 
EF1 has been generalized to \emph{weighted EF1 (WEF1)} \citep{ChakrabortyIgSu20} and PROP1 to \emph{weighted PROP1 (WPROP1)} \citep{AzizMoSa20}.
Yet, even though (weighted) envy-freeness implies (weighted) proportionality and EF1 implies PROP1, \citet{ChakrabortyIgSu20} have shown that, counterintuitively, WEF1 does not imply WPROP1.
Moreover, while it is always possible to satisfy each of the two notions separately, Chakraborty et al.~have demonstrated that it may be impossible to satisfy both simultaneously. 
In the ``share-based'' approach, 
 MMS has been generalized to \emph{weighted MMS (WMMS)} \citep{FarhadiGhHa19}.
A disadvantage of WMMS is that computing its value for an agent requires knowing not only the agent's own entitlement, but also the entitlement of every other agent. 
In addition, the definition of WMMS is rather difficult to understand and explain (especially when compared to MMS), thereby making it less likely to be adopted in practice.

In this paper, we expand and deepen our understanding of weighted fair division for indivisible items by generalizing existing weighted fairness notions, introducing new ones, and exploring the relationships between them.
Along the way, we propose solutions to the aforementioned issues of existing notions, and uncover a wide range of features of the weighted setting that are not present in the unweighted case typically studied in the fair division literature.

\subsection{Our contributions}
As is commonly done in fair division, we assume throughout the paper that agents have additive, nonnegative utilities.
In \textbf{Section~\ref{sec:fairness-up-to-1}}, we focus on fairness up to one item. 
The role of envy-freeness relaxations is to provide an upper bound on the amount of envy that is allowed between agents. 
In the unweighted setting, EF1 from agent $i$ towards agent~$j$ requires $i$'s envy to be at most $i$'s highest utility for an item in $j$'s bundle.
Formally, denoting $i$'s utility function by $u_i$ and $j$'s bundle by $A_j$, the envy should be at most $\max_{g\in A_j}u_i(g)$.
In the weighted setting, however,
envy is measured by comparing the scaled utilities $u_i(A_i)/w_i$ and $u_i(A_j)/w_j$, where $w_i$ and $w_j$ denote the agents' weights.
Therefore, 
one could reasonably argue that the amount of allowed envy should be similarly scaled to be either $u_i(g)/w_j$ or $u_i(g)/w_i$;
the first scaling corresponds to (hypothetically) removing the value of $g$ from $A_j$, while the second scaling corresponds to adding the value of $g$ to $A_i$.
Clearly, the first scaling yields a smaller envy allowance if and only if $w_i < w_j$, so it favors agents with smaller weights,
while the second scaling favors those with larger weights. 
We generalize both extremes at once
by defining the allowed envy to be a \emph{weighted average} of the two quantities, i.e.,  $u_i(g)\cdot \left(x/w_j + y/w_i\right)$ for $x+y = 1$. 
We denote this envy-freeness relaxation by WEF$(x,y)$.
Similarly, we define WPROP$(x,y)$ as a relaxation of proportionality.

For a fixed $x+y$, a higher~$x$ yields a stronger guarantee (e.g., low envy allowance) for lower-entitlement agents, while a lower $x$ yields a stronger guarantee for higher-entitlement agents.
We show that WEF$(x,y)$ implies WPROP$(x,y)$ for all $x,y$.
WEF1 corresponds to WEF$(1,0)$ while WPROP1 corresponds to WPROP$(0,1)$;
this provides an explanation as well as a resolution of the counterintuitive fact that WEF1 is incompatible with WPROP1.
With equal entitlements, WEF$(x,1-x)$ reduces to EF1 and WPROP$(x,1-x)$ reduces to PROP1 for all $x\in[0,1]$. 
But with different entitlements, the conditions obtained for different values of $x$ are different and incompatible.
For each $x\in[0,1]$, we show that WEF$(x,1-x)$, and therefore WPROP$(x,1-x)$, can always be satisfied by using a \emph{picking sequence}.\footnote{See the formal definition of a picking sequence in \Cref{sec:prelim}.
An open-source implementation is provided by \citet{han2024fast}.} 
On the other hand, for all distinct $x,x'\in[0,1]$, an allocation that fulfills both WPROP$(x,1-x)$ and WPROP$(x',1-x')$ may not exist, and an analogous impossibility result holds for WEF relaxations.

We also consider the \emph{maximum weighted Nash welfare (MWNW)} rule, which chooses an allocation that maximizes the weighted product of the agents' utilities.
Whereas every maximum Nash welfare allocation satisfies EF1 in the unweighted setting \citep{CaragiannisKuMo19},
we show that in the weighted setting, for each $x\in[0,1]$, there exists an instance in which no MWNW allocation satisfies WEF$(x,1-x)$ or even WPROP$(x,1-x)$.

Next, in \textbf{Section~\ref{sec:maximin}}, we explore share-based fairness notions. 
We propose a new weighted generalization of MMS called \emph{normalized MMS (NMMS)}, which is simply the MMS of an agent scaled by the agent's weight.
Not only is NMMS intuitive, but its definition also depends only on the relative entitlement of the agent in question (that is, the entitlement of the agent divided by the sum of all agents' entitlements), rather than on the entire list of entitlements.
We show that the best approximation factor that can always be attained for NMMS is $1/n$, matching the WMMS guarantee of \citet{FarhadiGhHa19}.
As with WMMS, the $1/n$-NMMS guarantee can be achieved using the (unweighted) round-robin algorithm, with the agents taking turns in decreasing order of their weights.
Interestingly, however, we show that WEF$(1,0)$ also implies $1/n$-NMMS---this means that a weighted round-robin algorithm as well as a generalization of \citet{BarmanKrVa18}'s market-based algorithm, which are known to ensure WEF$(1,0)$, also provide the optimal (worst-case) approximation of NMMS.
In addition, we establish relations between several fairness notions, thereby extending the work of \citet{AmanatidisBiMa18} to the weighted setting.
Some of the implications (or lack thereof) that we have derived are summarized in Table \ref{tab:results}, and existential results, both from our work and from prior work, are outlined in \Cref{table:existence}.

In \textbf{Section~\ref{sec:identical}}, we consider the case of identical items, also known as \emph{apportionment} \citep{BalinskiYo01}.
In this case, if the items were divisible, it would be clear that each agent should receive exactly the \emph{quota} proportional to her weight.
Hence, two important desiderata in apportionment are that each agent should obtain at least her \emph{lower quota} (the quota rounded down), and at most her \emph{upper quota} (the quota rounded up).
We show that WEF$(x,1-x)$ guarantees lower quota if and only if $x=0$ and upper quota if and only if $x=1$, whereas WMMS, NMMS, and WPROP$(x,1-x)$ for each $x\in [0,1]$ do not guarantee either quota.
Similarly, the MWNW rule does not always respect either quota.
Nevertheless, we introduce a new allocation rule, 
called \emph{weighted egalitarian (WEG)}, which aims to maximize the leximin vector of the agents' utilities,  with a carefully chosen normalization. 
We show that, in contrast to all other rules and notions studied herein, the WEG rule satisfies both quotas.
Moreover, in \textbf{Appendix~\ref{app:WEG-binary}}, we extend the study of WEG from identical items to binary additive valuations, and prove that in this domain, a WEG allocation gives every agent both her \emph{ordinal MMS (OMMS)} and \emph{AnyPrice share (APS)}. 
This indicates that the egalitarian approach may be worthy of further study in settings with different entitlements.

Since different picking sequences studied in Section~\ref{sec:fairness-up-to-1} satisfy WEF$(x,1-x)$ for different values of $x$, in \textbf{Section~\ref{sec:experiments}}, we perform experiments in order to compare them using common benchmarks.
In particular, we investigate the percentage of allocations produced by each picking sequence that satisfy weighted envy-freeness, weighted proportionality, WMMS-fairness, and NMMS-fairness.
Interestingly, we find that even though the picking sequences with medium values of $x$ perform best with respect to weighted envy-freeness, those with high values of $x$, which favor agents with low weights, fare better as far as weighted proportionality, WMMS-fairness, and NMMS-fairness are concerned.
Our findings demonstrate that the choice of fairness notion can play a decisive role when determining which rules or allocations are ``fairer'' than others.

\subsection{Related work}
\label{sec:related}

Fairness notions for agents with different entitlements were initially developed for settings with \emph{divisible} resources.
When dividing a continuous heterogeneous resource (also known as a \emph{cake}),
weighted proportional allocations always exist. Their properties and computational complexity have been investigated in several papers
\citep{barbanel1996game,Robertson1997Extensions,Shishido1999MarkChooseCut,dall2009disputed,Brams2011DivideandConquer,segalhalevi19cake,crew2020disproportionate,cseh2020complexity,janko2022cutting}.
Weighted envy-free allocations, which are also weighted proportional under additive utilities, always exist as well
\citep{Dubins1961How}
and can be computed in finite time \citep{Robertson1998CakeCutting,zeng2000approximate}.

When dividing \emph{homogeneous divisible resources} 
among agents with additive utilities,
an allocation that is weighted envy-free and satisfies the economic efficiency notion of \emph{Pareto optimality} is guaranteed to exist and can be found efficiently \cite{reijnierse1998finding,Aziz2014Cake}.
The setting with non-additive utilities corresponding to complementary items, typical for resource allocation problems in cloud computing environments, was studied by \citet{dolev2012no} and    \citet{gutman2012fair}.

\emph{Bargaining} is the abstract problem of selecting a feasible utility vector from the set of feasible utility vectors---fair division can be viewed as a special case of this problem.
Three classic bargaining solutions can be generalized to accommodate agents with different entitlements. 
In particular, \citet{kalai1977nonsymmetric} generalized the Nash bargaining solution by introducing the MWNW rule, while \citet{thomson1994cooperative} and \citet{driesen2012asymmetric} extended the Kalai--Smorodinsky bargaining solution and the leximin solution, respectively.
In Section~\ref{sub:weg}, we show that Driesen's extension does not work well in our setting with indivisible items, and present a different weighted extension of the leximin principle.

With indivisible items, the issue of fairness becomes more challenging, since weighted envy-free and weighted proportional allocations may not exist if all items must be allocated.
\citet{aziz2015fair}
studied a setting in which the agents reveal only their ordinal rankings over the items as opposed to their complete utility functions.
They presented a polynomial-time algorithm for checking whether there exists an allocation that is \emph{possibly weighted proportional} (i.e., weighted proportional according to \emph{at least one} utility profile consistent with the rankings), or 
\emph{necessarily weighted proportional} (i.e., weighted proportional according to \emph{all} utility profiles consistent with the rankings).

\citet{FarhadiGhHa19} introduced WMMS as a generalization of MMS to the setting of different entitlements. 
They showed that the best attainable multiplicative guarantee for WMMS-fairness is $1/n$ in general, and this guarantee can be improved to $1/2$ in the special case where the value that each agent has for every item is at most the agent's WMMS.
\citet{AzizChLi19} adapted the notion of WMMS to \emph{chores} (items that yield negative utilities).
They showed that, even with two agents, it is impossible to guarantee any ratio better than $4/3$ of the WMMS.\footnote{
Note that with chores, the approximation ratios are larger than $1$, and smaller ratios are better.
}
They then presented a $3/2$-WMMS approximation algorithm for two agents, and an exact WMMS algorithm for any number of agents with binary utilities (i.e., each agent has utility $0$ or $-1$ for each item).

WMMS is a \emph{cardinal} notion in the sense that, if the cardinal utilities of an agent change, then the set of bundles that provide WMMS for the agent may change even if the agent's ranking over bundles remains the same.
\citet{BabaioffNiTa21} implicitly introduced another generalization of MMS to the setting of different entitlements, which is based only on an agent's \emph{ordinal} ranking of the bundles---following \citet{Segalhalevi19}, we call it the \emph{ordinal maximin share (OMMS)}. 
Babaioff et al.~showed that OMMS-fairness can be attained by a competitive equilibrium with different budgets, where the budgets are proportional to the entitlements; however, such an equilibrium does not always exist.
The relations between ordinal MMS notions were further examined by \citet{Segalhalevi19,segal2020competitive}.
\citet{BabaioffEzFe21} proposed another ordinal notion, stronger than OMMS, which they called the \emph{AnyPrice share (APS)}. They gave a polynomial-time algorithm that guarantees every agent at least $3/5$ of her APS.

\citet{AzizMoSa20} presented a strongly polynomial-time algorithm that computes a Pareto optimal and WPROP1 (in our terminology, WPROP$(0,1)$) allocation for agents with different entitlements

\begin{landscape}
\renewcommand{\arraystretch}{1.2}
\begin{table}
\newcommand{\smallref}[1]{{\tiny (\ref{#1})}}
\centering
\begin{tabular}{|c|c|c|}
\hline
\textbf{Fairness notion}
& \textbf{Implies} 
& \textbf{Does not imply}
\\
\hline \hline

WEF$(x,1-x)$
& 
\begin{tabular}[c]{@{}c@{}}
WPROP$(x,1-x)$ (\ref{thm:wef_wprop}) \\
WWEF1 (\ref{thm:wwef1-wefxy}) \\
When $x=1$: $1/n$-NMMS (\ref{thm:wef10_nmms})
\end{tabular}
&
\begin{tabular}[c]{@{}c@{}}
When $x' \ne x$: WPROP$(x',1-x')$ (\ref{cor:WEF-WPROP-different}) \\
Any approximation of WMMS (\ref{thm:wef_wmms}) \\
When $x\ne 1$: Any approximation of NMMS (\ref{thm:wefx1mx_nmms}) \\
WEG (\ref{thm:wef_quota}, \ref{thm:weg_quota})
\end{tabular}
\\

\hline

WPROP$(x,1-x)$
& 
& 
\begin{tabular}[c]{@{}c@{}}
Any approximation of WMMS (\ref{cor:wpropx1mx_noapprox_wmms}) \\
Any approximation of NMMS (\ref{cor:wprop_nmms}) \\
OMMS, APS (\ref{thm:wprop_quota}, \ref{thm:aps_quota}) \\
MWNW, WEG (\ref{thm:wprop_quota}, \ref{cor:mwnw_quota_2agents}, \ref{thm:weg_quota})
\end{tabular}
\\

\hline

OMMS, APS 
& 
& 
\begin{tabular}[c]{@{}c@{}}
Any approximation of WMMS, NMMS (\ref{thm:aps_mms}) \\
WEG (\ref{thm:aps_quota}, \ref{thm:weg_quota})
\end{tabular}
\\

\hline

WMMS 
& $1/n$-NMMS  (\ref{thm:wmms_nmms})
& 
\begin{tabular}[c]{@{}c@{}}
Any approximation of OMMS, APS (\ref{thm:aps_mms}) \\
WEG (\ref{thm:mms_quota}, \ref{thm:weg_quota})
\end{tabular}
\\

\hline

NMMS 
& 
& 
\begin{tabular}[c]{@{}c@{}}
Any approximation of OMMS, APS (\ref{thm:aps_mms}) \\
Any approximation of WMMS (\ref{thm:nmms_wmms}) \\
WEG (\ref{thm:mms_quota}, \ref{thm:weg_quota})
\end{tabular}
\\

\hline

MWNW 
& 
& 
\begin{tabular}[c]{@{}c@{}}
WEF$(x,1-x)$, WPROP$(x,1-x)$ (\ref{thm:mwnw_wprop}) \\
Any approximation of WMMS (\ref{thm:mwnw_wmms}) \\
Any approximation of NMMS (\ref{thm:mwnw_nmms}) \\
WEG (\ref{thm:mwnw_quota}, \ref{thm:weg_quota})
\end{tabular}
\\
\hline

WEG 
& 
With binary utilities: APS and OMMS  (\ref{cor:weg_aps})
& 
\begin{tabular}[c]{@{}c@{}}
WEF$(x,1-x)$ (\ref{rem:weg}) \\
Any approximation of WMMS, NMMS (\ref{rem:weg})
\end{tabular}
\\
\hline

OEF1
& 
\begin{tabular}[c]{@{}c@{}}
$1/n$-WMMS \cite{FarhadiGhHa19} \\
$1/n$-NMMS (\ref{thm:oef1_nmms})
\end{tabular}
&

\\
\hline
\end{tabular}
\vspace{2mm}
\caption{
\label{tab:results}
Summary of (non-)implications between weighted fairness notions.
OMMS/APS and OEF1 are defined in Sections~\ref{sec:share-definitions} and \ref{sec:OEF1}, respectively.
}
\end{table}
\end{landscape}

\renewcommand{\arraystretch}{1.2}
\begin{table}[!ht]
\centering
\begin{tabular}{| c | c |}
\hline
\textbf{Notion} & \textbf{Guaranteed existence} \\ \hline \hline 
WEF$(x,1-x)$ & Any $x\in [0,1]$ (\Cref{thm:pickseq})
  \\ \hline
WPROP$(x,1-x)$ & Any $x\in [0,1]$ (\Cref{thm:pickseq} and \Cref{thm:wef_wprop}) \\ \hline
$\alpha$-WMMS & $\alpha = 1/n$, and not for any $\alpha > 1/n$ \citep{FarhadiGhHa19} \\ \hline
$\alpha$-NMMS & $\alpha = 1/n$, and not for any $\alpha > 1/n$ (\Cref{thm:wef10_nmms}) \\ \hline
$\alpha$-OMMS & 
\begin{tabular}[c]{@{}c@{}}
$\alpha = 3/5$, and not for any $\alpha > 39/40$ 
\\ \citep{Segalhalevi19,BabaioffEzFe21,FeigeSaTa21}
\end{tabular} 
\\ \hline
$\alpha$-APS & 
\begin{tabular}[c]{@{}c@{}}
$\alpha = 3/5$, and not for any $\alpha > 39/40$ 
\\ \citep{BabaioffEzFe21,FeigeSaTa21}
\end{tabular}
 \\ \hline
\end{tabular}
\vspace{2mm}
\caption{Summary of results on the guaranteed existence of allocations satisfying (approximations of) weighted fairness notions.
}
\label{table:existence}
\end{table}

\noindent and mixed (i.e., positive and negative) utilities.
\citet{ChakrabortyIgSu20} introduced WEF1 (in our terminology, WEF$(1,0)$) and showed that this notion can be attained by a weighted version of the round-robin algorithm.
\citet{ChakrabortyScSu21} considered a larger class of picking sequences based on apportionment methods and showed that all of them satisfy a relaxation of WEF1 called \emph{weak WEF1 (WWEF1)}; they also demonstrated that these picking sequences perform well in relation to monotonicity properties.
\citet{SuksompongTe22} studied MWNW for binary (additive) valuations with respect to monotonicity and strategyproofness; the same authors later extended their results to rules maximizing other welfare notions and binary submodular valuations \citep{SuksompongTe23}.

When \emph{monetary transfers} are allowed, exact fairness is attainable. 
\citet{corradi2001adjusted} defined an allocation to be \emph{equitable} if the utility of each agent $i$ (defined as the value of items plus the money given to $i$) is $r\cdot w_i$ times $i$'s utility for the entire set of items, where $r$ is the same constant for all agents.
They presented an algorithm, generalizing the \emph{adjusted Knaster} algorithm \citep{raith2000fair}, that finds an equitable allocation with $r\geq 1$; this allocation is also weighted proportional (but not necessarily weighted envy-free).

After the conference version of our work was published, a number of authors further studied and applied our notions WEF$(x,y)$ and WPROP$(x,y)$.
\citet{AzizGaMi23} and \citet{HoeferScVa23} investigated \emph{best-of-both-worlds} fairness, where the goal is to find random allocations that are fair both before the randomization (\emph{ex-ante}) as well as after (\emph{ex-post}).
Interestingly, they showed that while ex-ante WEF is always compatible with ex-post WEF$(1,1)$, it may be incompatible with ex-post WEF$(x,y)$ whenever $x,y\le 1$ and $x+y < 2$.
\citet{MontanariScSu24} went beyond additive valuations and proposed generalizations of WEF$(x,1-x)$ that can be satisfied under submodular valuations.
\citet{WuZhZh23} examined weighted envy-freeness in chore division and proved that, as is the case with goods, a WEF$(x,1-x)$ allocation is guaranteed to exist for every $x$.

Finally, work on weighted fair allocation of indivisible items was summarized in a recent survey by \citet{Suksompong25}.

\section{Preliminaries}
\label{sec:prelim}

We consider a setting with a set $N = [n]$ of agents and a set $M=[m]$ of items, where $n\ge 2$ and $[k] := \{1,2,\dots,k\}$ for each positive integer $k$.
A (possibly empty) subset of items is called a \emph{bundle}.
The \emph{entitlement} or \emph{weight} of each agent $i\in N$ is denoted by $w_i > 0$. 
For each subset of agents $N'\subseteq N$, we denote $w_{N'} := \sum_{i\in N'}w_i$.
Even though only relative entitlements (i.e., the ratios $w_i/w_N$ for $i\in N$) matter for our paper---and we could therefore technically normalize $w_N$ to $1$---we will keep the term $w_N$ flexible, as this flexibility will sometimes be useful for our exposition.

An \emph{allocation} $A=(A_1, A_2, \dots, A_n)$ is an ordered partition of $M$ into $n$ bundles such that bundle $A_i$ is assigned to agent $i$. 
(This is sometimes called a \emph{complete allocation} by other authors, that is, all items must be allocated.)
Each agent $i\in N$ has a \emph{utility function} $u_i$, which we assume to be nonnegative and additive.
This means that for each $M'\subseteq M$, $u_i(M') = \sum_{g\in M'}u_i(\{g\})$.
For simplicity, we will sometimes write $u_i(g)$ instead of $u_i(\{g\})$ for $g\in M$. 

An allocation satisfies \emph{envy-freeness up to one item (EF1)} if for every pair of agents $i,j\in N$, there exists $B\subseteq A_j$ with $|B| \le 1$ such that\footnote{We use a set $B\subseteq A_j$ with $|B|\le 1$ instead of an item $g\in A_j$ in order to handle the case where $A_j = \emptyset$.} $u_i(A_i) \ge u_i(A_j\setminus B)$.
It satisfies \emph{proportionality up to one item (PROP1)} if for every agent $i\in N$, there exists $B\subseteq M\setminus A_i$ such that $u_i(A) \ge u_i(M)/n - u_i(B)$.

A \emph{picking sequence} is a sequence $(p_1,\dots,p_m)$ where $p_j\in N$ for each $j$. 
The \emph{output} of a picking sequence is the allocation resulting from the process in which, 
in the $j$-th turn, agent $p_j$ picks her favorite item from the remaining items, breaking ties in a consistent manner (for example, in favor of lower-numbered items).

\section{Fairness Up to One Item}
\label{sec:fairness-up-to-1}

In this section, we define the envy-freeness and proportionality relaxations WEF$(x,y)$ and WPROP$(x,y)$, and study them particularly for the case where $x+y = 1$.

\subsection{Weighted envy-freeness notions}
An allocation is called \emph{weighted envy-free (WEF)} if for every pair $i,j\in N$,
\begin{align*}
\frac{u_i(A_i)}{w_i} \geq \frac{u_i(A_j)}{w_j}.
\end{align*}
We define a continuum of weighted envy-freeness relaxations parameterized by two nonnegative real values, $x$ and $y$. 
\begin{definition}[\wefxy]
For $x,y\in [0,1]$, an allocation $(A_1,\dots,A_n)$ satisfies \emph{WEF$(x,y)$} if for every pair $i,j\in N$, there exists $B\subseteq A_j$ with $|B| \le 1$ such that
\begin{align*}
&
\frac{u_i(A_i) + y\cdot u_i(B)}{w_i} \ge \frac{u_i(A_j) - x\cdot u_i(B)}{w_j},
\end{align*}
or equivalently,
\begin{align*}
\frac{u_i(A_j)}{w_j} - \frac{u_i(A_i)}{w_i} 
\leq 
\bigg(
\frac{y}{w_i} + \frac{x}{w_j}
\bigg)\cdot
u_i(B).
\end{align*}
The difference on the left-hand side, when non-negative, corresponds to the \emph{weighted envy} of $i$ towards $j$. 
The definition requires this quantity to be upper-bounded by the value of at most one item in $j$'s bundle, 
scaled by a weighted average of $1/w_i$ and $1/w_j$.
\end{definition}

\wefxy generalizes and interpolates between several previously studied notions. 
In particular, 
WEF$(0,0)$ is the same as weighted envy-freeness, 
WEF$(1,0)$ is equivalent to the notion WEF1 proposed by \citet{ChakrabortyIgSu20}, 
and WEF$(1,1)$ corresponds to what these authors called ``transfer weighted envy-freeness up to one item''.
(This is also why we chose to allow any $x,y\in [0,1]$ rather than impose the condition $x+y \le 1$.)
Chakraborty et al.~defined the \emph{weak WEF1 (WWEF1)} condition as follows:
for each pair $i,j$ of agents, $i$ does not envy $j$ after either (i) at most one item is removed from $A_j$, or (ii) at most one item from $A_j$ is copied over to $A_i$.
(Note that (i) corresponds to the WEF$(1,0)$ condition and (ii) to the WEF$(0,1)$ condition.)
The following two theorems show that, for every $x\in[0,1]$, WWEF1 is a strictly weaker condition than WEF$(x,1-x)$.

\begin{theorem}
\label{thm:wwef1-wefxy}
WWEF1 is equivalent to the following condition:
for every pair of agents $i,j$,
there exists some $x_{i,j} \in [0,1]$ such that the condition for WEF$(x_{i,j},1-x_{i,j})$ is satisfied for that pair.
\end{theorem}
\begin{proof}
One direction is obvious: in a WWEF1 allocation, by definition, for every pair $i,j$, the allocation is either WEF$(1,0)$ or WEF$(0,1)$, so it satisfies the condition for some $x_{i,j}\in\{0,1\}$.

For the other direction, recall that WEF$(x_{i,j},1-x_{i,j})$ means that, for some $B\subseteq A_j$ with $|B|\leq 1$,
\begin{align*}
\frac{u_i(A_j)}{w_j} - \frac{u_i(A_i)}{w_i} 
&\leq 
\bigg(
\frac{1-x_{i,j}}{w_i} 
+
\frac{x_{i,j}}{w_j}
\bigg)\cdot
u_i(B)
\\
&
=
\bigg(
\frac{1}{w_i} 
+
x_{i,j}\cdot
\left(
\frac{1}{w_j}
-
\frac{1}{w_i} 
\right)
\bigg)\cdot 
u_i(B).
\end{align*}
If $w_i>w_j$, then the same inequality holds if we replace $x_{i,j}$ with $1$;
otherwise, the same inequality holds if we replace $x_{i,j}$ with $0$.
Indeed, this implication follows from analyzing the sign of $1/w_j - 1/w_i$.
\end{proof}

Theorem \ref{thm:wwef1-wefxy} shows that, for every $x\in[0,1]$, WEF$(x,1-x)$ implies WWEF1: the former fixes the same $x$ for all pairs of agents, while the latter allows a different $x$ for each pair of agents.
The next theorem shows that the implication is strict.

\begin{theorem}
\label{thm:WWEF1-WEFxy}
For each $x\in[0,1]$, WWEF1 does not imply WEF$(x,1-x)$.
\end{theorem}

Theorem~\ref{thm:WWEF1-WEFxy} will later be strengthened by Corollary~\ref{cor:WWEF1-WPROP}, so we do not present its proof here.

As discussed, WEF$(x,y)$ requires the weighted envy of $i$ towards~$j$, i.e., $\max\left\{0,\frac{u_i(A_j)}{w_j} - \frac{u_i(A_i)}{w_i}\right\}$, to be at most $\left(\frac{y}{w_i}+\frac{x}{w_j}\right)\cdot u_i(B)$.
It is evident that, 
with equal entitlements, this condition
depends only on the sum $x+y$; in particular, whenever $x+y=1$, WEF$(x,y)$ is equivalent to EF1.
However, with different entitlements, every selection of $x,y$, even with $x+y=1$, leads to a different condition: a higher $x$ yields a stronger guarantee (i.e., lower allowed weighted envy) for agent~$i$ when $w_i < w_j$, while a higher $y$ yields a stronger guarantee for the agent when $w_i > w_j$.
This raises two natural questions.
First, for which pairs $(x,y)$ can WEF$(x,y)$ be guaranteed?
Second, is it possible to guarantee WEF$(x,y)$ for multiple pairs $(x,y)$ simultaneously?

For the first question, whenever $x+y < 1$, a standard example of two agents with equal weights and one valuable item shows that \wefxy cannot always be satisfied.
We will show next that \wefxy{} can always be satisfied when $x+y=1$, thereby implying existence for $x+y > 1$ as well. 
To this end, we characterize picking sequences whose output is guaranteed to satisfy WEF$(x,1-x)$; this generalizes the WEF$(1,0)$ characterization of \citet{ChakrabortyScSu21}.
For brevity, we say that a picking sequence satisfies a fairness notion if its output always satisfies that notion.
The proof of this theorem is similar to that of \citet{ChakrabortyScSu21} and therefore deferred to Appendix~\ref{app:omitted}.

\begin{theorem}
\label{thm:picking-WEF}
Let $x\in[0,1]$.
A picking sequence $\pi$ satisfies WEF$(x,1-x)$ if and only if for every prefix $P$ of $\pi$ and every pair of agents $i,j$, where agent $i$ has $t_i$ picks in $P$
and agent~$j$ has $t_j$ picks in $P$,
we have 
$t_i+(1-x)\ge \frac{w_i}{w_j}\cdot(t_j-x)$.
\end{theorem}

We can now prove that a WEF$(x,1-x)$ allocation exists in every instance.

\begin{theorem}
\label{thm:pickseq}
Let $x\in [0,1]$.
Consider a picking sequence $\pi$ such that in each turn, the pick is assigned to an agent~$i$ with the smallest $\frac{t_i+(1-x)}{w_i}$, where $t_i$ is the number of times agent $i$ has picked so far.
Then, $\pi$ satisfies WEF$(x,1-x)$.
\end{theorem}

\begin{proof}
Consider any pair of agents $i,j$ who have picked $t_i$ and $t_j$ items, respectively.
Suppose first that $t_j\geq 1$. Since $j$ was allowed to pick an item when $j$ had $t_j-1$ picks, it must be that $\frac{(t_j-1)+(1-x)}{w_j} \le \frac{t_i+(1-x)}{w_i}$, which is equivalent to  $t_i+(1-x) \ge \frac{w_i}{w_j}\cdot(t_j-x)$.
The same condition obviously holds if $t_j=0$.
Hence, by Theorem~\ref{thm:picking-WEF}, the allocation is WEF$(x,1-x)$.
\end{proof}

In particular, the picking sequence in \Cref{thm:pickseq} with $x = 1/2$, which reduces to Webster's apportionment method when the items are identical \citep[p.~100]{BalinskiYo01},
satisfies WEF$(1/2,1/2)$. 
This strengthens a result of \citet{ChakrabortyScSu21} that this picking sequence guarantees WWEF1.

While every WEF$(x,1-x)$ notion can be satisfied on its own, it is impossible to guarantee WEF$(x,1-x)$ for two different values of $x$ simultaneously.
We prove this strong incompatibility result even for the weaker notion of WPROP$(x,y)$, which we define next.

\subsection{Weighted proportionality notions}
An allocation is called \emph{weighted proportional (WPROP)} if for every pair $i,j\in N$,
\begin{align*}
\frac{u_i(A_i)}{w_i} \ge \frac{u_i(M)}{w_N}.
\end{align*}

Similarly to WEF$(x,y)$, we define a continuum of weighted proportionality relaxations.
\begin{definition}[\wpropxy]
\label{def:wpropxy}
For $x,y\in [0,1]$, an allocation $(A_1,\dots,A_n)$ satisfies \emph{WPROP$(x,y)$} if for every $i\in N$, there exists $B\subseteq M\setminus A_i$ with $|B| \le 1$ such that 
\begin{align*}
\frac{u_i(A_i) + y\cdot u_i(B)}{w_i}
& \ge 
 \frac{u_i(M) - n\cdot x\cdot u_i(B)}{w_N},
\end{align*}
or equivalently,
\begin{align*}
\frac{u_i(M)}{w_N}  - \frac{u_i(A_i)}{w_i}
\leq  \left(\frac{n\cdot x}{w_N} 
+ \frac{y}{w_i}\right) \cdot u_i(B).
\end{align*}
\end{definition}
The expression on the left-hand side, when non-negative, is the amount by which agent $i$'s bundle falls short of her weighted proportionality threshold. 
The definition requires this quantity to be upper-bounded by the value of at most one item not in $i$'s bundle, 
scaled by a weighted average of $n/w_N$ and $1/w_i$.
The factor $n$  can be thought of as a normalization factor, since $B$ is removed from the entire set of items $M$, which is distributed among $n$ agents, rather than from another agent's bundle as in the definition of \wefxy.
This normalization ensures that, similarly to the \wefxy family, when all entitlements are equal, \wpropxy reduces to PROP1 whenever $x+y = 1$.

WPROP$(0,0)$ is the same as weighted proportionality, while WPROP$(0,1)$ is equivalent to the notion WPROP1 put forward by \citet{AzizMoSa20}.
With different entitlements and $x+y = 1$, 
a higher~$x$ yields a stronger guarantee for agent~$i$ (i.e., the agent's utility cannot be far below her proportional share) when $w_i < w_N/n$, while a higher $y$ yields a stronger guarantee for the agent when $w_i > w_N/n$
(note that the quantity $w_N/n$ is the average weight of the $n$ agents).
This raises the same two questions that we posed for WEF$(x,y)$: 
For which pairs $(x,y)$ can WPROP$(x,y)$ always be attained? 
And is it possible to attain it for different pairs $(x,y)$ simultaneously?

To answer these questions, we first define a stronger version of \wpropxy which we call \wpropstarxy. 

\begin{definition}[\wpropstarxy]\label{def:strongwpropxy}
For $x,y\in [0,1]$, an allocation $(A_1,\dots,A_n)$
satisfies \wpropstarxy if for each $i\in N$, the following holds:
There exists $B\subseteq M\setminus A_i$ with $|B|\le 1$ and, for every $j\in N\setminus \{i\}$, there exists $B_j\subseteq A_j$ with $|B_j| \le 1$ such that 
\begin{align*}
\frac{u_i(A_i) + y\cdot u_i(B)}{w_i} \ge 
 \frac{u_i(M) -  x\cdot \sum_{j \in N \setminus \{i\}} u_i(B_j)}{w_N}.
\end{align*}
\end{definition}

Since $\sum_{j \in N \setminus \{i\}}u_i(B_j) \le (n-1)\cdot u_i(B) \le  n\cdot u_i(B)$,
where we take $B$ to be the singleton set containing $i$'s most valuable item in $M\setminus A_i$ if $A_i\neq M$ and the empty set otherwise,
\wpropstarxy is a strengthening of \wpropxy.
Nevertheless, it is less intuitive than \wpropxy. 
As such, we do not propose \wpropstarxy as a major fairness desideratum in its own right, but it will turn out to be a useful concept for establishing some of our results.

We first show that for all $x$ and $y$, WEF$(x,y)$ implies \wpropstarxy, which in turn implies \wpropxy; this generalizes the fact that weighted envy-freeness implies weighted proportionality, which corresponds to taking $x=y=0$.
Note that while we generally think of $x$ and $y$ as being constants not depending on $n$ or $w_1,\dots,w_n$, in the following two implications we will derive more refined bounds in which $x$ and $y$ may depend on these parameters.

\begin{lemma}
\label{lem:WEF-implies-strongWPROP}
For all real numbers $x,y\in[0,1]$ and $y' \ge (1-\frac{\min_{k\in N} w_k}{w_N})y$, 
\wefxy implies \wpropstar{(x,y')}.

In particular, WEF$(x,y)$ implies \wpropstarxy.
\end{lemma}

\begin{proof}
Consider a WEF$(x,y)$ allocation $A$, and fix an arbitrary agent $i\in N$.
If $A_i = M$, the \wpropstarxy condition is trivially satisfied for $i$, so assume that $A_i\subsetneq M$.
For every $j\in N \setminus \{i\}$, the \wefxy condition implies that for some $B_j\subseteq A_j$ of size at most $1$, we have
\begin{align*}
w_j\cdot u_i(A_i) 
&\ge w_i\cdot u_i(A_j) - w_i x \cdot u_i(B_j) - w_j y\cdot u_i(B_j).
\end{align*}
Let $B$ be a singleton set containing $i$'s most valuable item in $M\setminus A_i$.
Summing this inequality over all $j\in N\setminus\{i\}$ and adding $w_i\cdot u_i(A_i)$ to both sides, 
we get

\begin{align*}
w_N\cdot u_i(A_i) 
&\ge  w_i\cdot\sum_{j\in N} u_i(A_j) - w_i x\cdot \sum_{j\in N\setminus\{i\}} u_i(B_j)  -y \cdot\sum_{j\in N\setminus\{i\}}w_ju_i(B_j) 
\\
&\ge
w_i  u_i(M) - w_i x \cdot\sum_{j\in N\setminus\{i\}}  u_i(B_j) -y \cdot\sum_{j\in N\setminus\{i\}}w_ju_i(B) 
\\
&=
w_i  u_i(M) - w_i x \cdot\sum_{j\in N\setminus\{i\}}  u_i(B_j)  -  u_i(B)\cdot (w_N - w_i)y.
\end{align*}
Now, 
\[(w_N -  w_i) y \le (w_N -  \min_{k\in N}w_k) y \le w_N y'.\]
Combining this inequality with the inequality before it and rearranging yields
\begin{align*}
w_N\cdot u_i(A_i) + u_i(B)\cdot w_N y' 
\geq
w_i  u_i(M) - w_i x \cdot\sum_{j\in N\setminus\{i\}}  u_i(B_j),
\end{align*}
or equivalently,
\begin{align*}
u_i(A_i) + y' \cdot u_i(B) \ge \frac{w_i}{w_N}\left( u_i(M) - x\cdot \sum_{j\in N\setminus\{i\}}  u_i(B_j) \right).
\end{align*}
Hence, the WPROP$^*(x,y')$ condition is satisfied for agent~$i$.

The second part of the lemma follows from the fact that $y \ge (1-\frac{\min_{k\in N} w_k}{w_N}) y $.
\end{proof}

\begin{lemma}
\label{lem:strongwprop_wprop}
For all real numbers $x,y\in[0,1]$ and $x' \ge \left(1-\frac{1}{n}\right)x $, \wpropstarxy implies WPROP$(x',y)$.

In particular, \wpropstarxy implies \wpropxy.
\end{lemma}

\begin{proof}
Fix a \wpropstarxy allocation $A$.
It is sufficient to prove that the 
WPROP$(x',y)$ condition holds for an arbitrary agent $i\in N$.
If $A_i = M$, the WPROP$(x',y)$ condition is trivially satisfied for $i$, so assume that $A_i\subsetneq M$.
By definition of \wpropstarxy, there exist $B\subseteq M\setminus A_i$ of size at most $1$ and, for each $j\in N\setminus\{i\}$, $B_j\subseteq A_j$ of size at most $1$ such that
\begin{align*}
\frac{u_i(A_i) + y\cdot u_i(B)}{w_i} \ge 
 \frac{u_i(M) -  x\cdot \sum_{j \in N \setminus \{i\}} u_i(B_j)}{w_N}.
\end{align*}
In particular, this holds when $B$ is a singleton set containing $i$'s most valuable item in $M\setminus A_i$.
Since $u_i(B_j) \le u_i(B)$ for all $j\in N\setminus\{i\}$, we have
\begin{align*}
\frac{u_i(A_i) + y\cdot u_i(B)}{w_i} 
&\ge  \frac{u_i(M) -  x\cdot \sum_{j \in N \setminus \{i\}} u_i(B)}{w_N} \\
&= \frac{u_i(M) - (n-1)x\cdot u_i(B)}{w_N} \\
&\ge \frac{u_i(M) - nx'\cdot u_i(B)}{w_N}.
\end{align*}
Hence, the WPROP$(x',y)$ condition is satisfied for agent~$i$.

The second part of the lemma follows from the fact that $x \ge (1-\frac{1}{n})x$.
\end{proof}

\begin{corollary}
\label{thm:wef_wprop}
For all $x,y\in[0,1]$, 
\wefxy implies \wpropxy.
\end{corollary}

Combining \Cref{thm:wef_wprop} and  \Cref{thm:pickseq}, we find that WPROP$(x,1-x)$ can be guaranteed for each fixed $x$. 
Can it be guaranteed for two different values of $x$ simultaneously? 
Can it be guaranteed for $x,y$ such that $x+y<1$? 
The following theorem (along with the discussion after it) shows that the answer to both questions is negative. 

\begin{theorem}
\label{thm:WPROP-simultaneous}
For all $x,x',y,y'\in[0,1]$, if $x'+y<1$ or $x+y'<1$,
there is an instance with identical items in which no allocation is both WPROP$(x,y)$ and WPROP$(x',y')$.
\end{theorem}

\begin{proof}
We handle the case $x'+y < 1$; the case $x+y'<1$ is analogous.
Consider an instance with $m$ identical items.
The entitlements are $w_1 = \dots = w_{n-1} = \varepsilon $ and $w_n = 1 - (n-1) \varepsilon $ for some small $\varepsilon > 0$.
We choose $m, n, \varepsilon$ satisfying the following two inequalities (we will show later that such a choice is possible):
\begin{align*}
\text{(*)}&&\varepsilon  m - \varepsilon  n  x - y &> 0;
\\
\text{(**)}&&
(1 - (n-1)\varepsilon ) (m - nx') - y' &> m-n+1.
\end{align*}
WPROP$(x,y)$ requires that, for each of the first $n-1$ agents,
\begin{align*}
\frac{u_i(A_i) + y}{\varepsilon} &\geq \frac{m - nx}{1},
\end{align*}
which is equivalent to \[u_i(A_i) \geq \varepsilon m - \varepsilon n x - y > 0\] by (*),
so each of these agents must get at least one item.
WPROP$(x',y')$ requires that, for the last agent,
\begin{align*}
\frac{u_n(A_n) + y'}{1-(n-1)\varepsilon} &\geq \frac{m - n x'}{1},
\end{align*}
which is equivalent to \[u_n(A_n) \geq (1-(n-1)\varepsilon) (m - n x') - y' > m-n+1\] by (**),
so agent $n$ must get at least $m-n+2$ items. 
However, this is infeasible, as there are only $m$ items in total.

We now show that it is possible to choose $\varepsilon,m,n$ satisfying (*) and (**).
Inequality (*) can be rearranged to $\varepsilon m > \varepsilon n x + y$.
Inequality (**), upon dividing by $n-1$ and rearranging, becomes 
\begin{align*}
\varepsilon m < 1 + \varepsilon n x' - \frac{y' + x' n}{n-1}.
\end{align*}
Hence, as long as the gap between the quantities on the right-hand sides of the two rearranged inequalities is greater than $\varepsilon$, there exists an $m$ satisfying both inequalities.
So it is sufficient to find $n, \varepsilon$ such that 
\begin{align*}
(\varepsilon n x + y) + \varepsilon
<
1 + \varepsilon  n x' - \frac{y' + x'n}{n-1}.
\end{align*}
Rearranging, this becomes 
\begin{align*}
\varepsilon \cdot (nx + 1 - nx')
<
1 - y - \frac{y' + x'n}{n-1}.
\end{align*}
The limit as $n \to \infty$ of the right-hand side is $1-y-x'$, which is positive by assumption. 
Choose $n$ large enough so that the right-hand side is positive.
Then, choose $\varepsilon$ small enough so that the left-hand side is smaller than the right-hand side, and therefore the inequality holds.
\end{proof}

\Cref{thm:WPROP-simultaneous} has several corollaries.

First, taking $y = 1-x$ and $y' = 1-x'$ yields that, if $x\neq x'$,
there is an instance in which no allocation is both WPROP$(x,1-x)$ and WPROP$(x',1-x')$.
Combining this with \Cref{thm:wef_wprop} gives the same incompatibility for
WEF$(x,1-x)$ and WEF$(x',1-x')$.

Second, taking $x' = x$ and $y' = y$ yields that, when $x + y < 1$, a WPROP$(x,y)$ allocation is not guaranteed to exist,
and therefore a WEF$(x,y)$ allocation may not exist.

Third, we can generalize \citet{ChakrabortyIgSu20}'s result that WEF$(1,0)$ does not imply WPROP$(0,1)$:

\begin{corollary}
\label{cor:WEF-WPROP-different}
For all distinct $x,x'\in[0,1]$, WEF$(x,1-x)$ does not necessarily imply WPROP$(x',1-x')$.
\end{corollary}

\begin{proof}
Assume for contradiction that WEF$(x,1-x)$ implies WPROP$(x',1-x')$ for some $x\ne x'$.
By \Cref{thm:pickseq}, a WEF$(x,1-x)$ allocation always exists, and therefore such an allocation is WPROP$(x',1-x')$.
By \Cref{thm:wef_wprop}, this allocation is also WPROP$(x,1-x)$.
Hence, an allocation that is both WPROP$(x,1-x)$ and WPROP$(x',1-x')$ always exists, contradicting the first corollary of \Cref{thm:WPROP-simultaneous} mentioned above.
\end{proof}

Similarly to \Cref{thm:picking-WEF} for WEF$(x,1-x)$, we next characterize picking sequences whose output is guaranteed to satisfy WPROP$(x,1-x)$, thereby generalizing the characterization of \citet{ChakrabortyScSu21} for WPROP$(0,1)$.
The proof of this result is deferred to \Cref{app:omitted}.

\begin{theorem}
\label{thm:picking-WPROP}
Let $x\in[0,1]$.
A picking sequence $\pi$ satisfies WPROP$(x,1-x)$ if and only if for every prefix of $\pi$ and every agent $i$, we have $t_i + (1-x)\ge \frac{w_i}{w_N}\cdot (k-nx)$, where $t_i$ and $k$ denote the number of agent $i$'s picks in the prefix and the length of the prefix, respectively.
\end{theorem}

The picking sequences in \Cref{thm:pickseq} can be derived from \emph{divisor methods}, a well-studied class of methods in apportionment \citep{BalinskiYo01}.
Each divisor method gives rise to a picking sequence associated with a function $f:\mathbb{Z}_{\ge 0} \rightarrow \mathbb{R}_{\ge 0}$ with $t\le f(t) \le t+1$ for each $t$, and assigns each pick to an agent~$i$ who minimizes the ratio $f(t_i)/w_i$, where $t_i$ denotes the number of times that agent~$i$ has picked so far.
We will refer to such picking sequences as \emph{divisor picking sequences}.
We show next that if we restrict our attention to divisor picking sequences, then for each $x$, the picking sequence in \Cref{thm:pickseq} is the only one satisfying either WEF$(x,1-x)$ or WPROP$(x,1-x)$.
This generalizes the result of \citet{ChakrabortyScSu21} for WEF$(1,0)$ and WPROP$(0,1)$.

\begin{theorem}
\label{thm:divisor}
Let $x\in[0,1]$ be a real number.
A divisor picking sequence with function~$f$ satisfies WPROP$(x,1-x)$ if and only if $f(t) = t+(1-x)$ for all $t\in\mathbb{Z}_{\ge 0}$.
The same holds for WEF$(x,1-x)$.
\end{theorem}

\begin{proof}
Let $y := 1-x$.
By \Cref{thm:wef_wprop} and \Cref{thm:pickseq}, it suffices to prove that if a divisor picking sequence with function $f$ satisfies WPROP$(x,y)$, then $f(t) = t+y$ for all $t\in\mathbb{Z}_{\ge 0}$.
We prove this implication even for the special case in which all items are identical; in this case, WPROP$(x,y)$ implies that the number of items given to each agent $i$ must be at least $\frac{w_i}{w_N}\cdot (m-nx)-y$.

We first show that  
\begin{align}
\label{eq:fa-fb}
\frac{f(a)}{f(b)} \leq \frac{a+y}{b+y}
\text{~~~~~for all integers~~~}
a\ge 0,~ b \ge 1.
\end{align}
Fix $a$ and $b$, and assume for contradiction that $\frac{f(a)}{f(b)} > \frac{a+y}{b+y}$.
Choose $w_1$ and $w_2$ such that $\frac{f(a)}{f(b)} > \frac{w_1}{w_2} > \frac{a+y}{b+y}$, choose $n$ large enough so that 
$\frac{w_1}{w_2} > \frac{a+y}{b+y-\frac{1}{n-1}}$, and consider an instance with $a+(n-1)(b+1)$ identical items and $n$ agents such that $w_2 = w_3 = \dots = w_n$.
Since $\frac{f(a)}{w_1} > \frac{f(b)}{w_2}$, 
agents $2,\ldots,n$ all pick their $(b+1)$st item before agent $1$ picks her $(a+1)$st item; therefore, 
agent~$1$ gets at most $a$ items overall.
Moreover, since $\frac{w_1}{w_2} > \frac{a+y}{b+y-\frac{1}{n-1}}$, we have
\begin{align*}
((n-1)(b+y)-1)w_1 > (a+y)(n-1)w_2,
\end{align*}
that is,
\begin{align*}
w_1+(a+y)(n-1)w_2 &< w_1(n-1)(b+y).
\end{align*}
This implies that
\[
(a+y)(w_1+(n-1)w_2) < w_1(a+y-1+(n-1)(b+y)),
\]
or equivalently,
\[
a < \frac{w_1}{w_1+(n-1)w_2}\cdot(a+(n-1)(b+1)-nx) - y
=
\frac{w_1}{w_N}\cdot(m-nx) - y,
\]
which violates the WPROP$(x,y)$ condition for agent~$1$.
It follows that $\frac{f(a)}{f(b)} \le \frac{a+y}{b+y}$, as claimed in~\eqref{eq:fa-fb}.

We now show that $f(t) = t+y$ for all $t \ge 1$.
For each positive integer $k$, 
we can apply~\eqref{eq:fa-fb} for both $(t,kt)$ and the inverse pair $(kt,t)$. This implies
$
\frac{f(kt)}{f(t)} = \frac{kt+y}{t+y}
$, which in turn implies $f(kt) = \frac{kt+y}{t+y}\cdot f(t)$.
If $f(t) > t+y$, say $f(t) = t+y+\varepsilon$ for some $\varepsilon > 0$, then $f(kt) = kt+y+\frac{kt+y}{t+y}\cdot\varepsilon$; the latter quantity exceeds $kt+1$ when $k$ is sufficiently large, contradicting the definition of divisor picking sequences.
The case $f(t) < t+y$ leads to a similar contradiction.
This means that $f(t) = t+y$ for all $t\ge 1$.

It remains to show that $f(0) = y$.
We consider two cases.

\underline{Case 1}: $y = 0$. 
Applying \eqref{eq:fa-fb}
with $a=0$ and $b=1$ yields $\frac{f(0)}{f(1)} \le \frac{0}{1}$, which  implies $f(0) = 0$.

\underline{Case 2}: $y > 0$. 
First, assume for contradiction that $f(0) = 0$.
We will construct a new instance with $m = n$ identical items, each with value $1$.
Choose $n$ large enough so that $\frac{1+y}{ny} < 1$, and choose $w_1$ large enough so that $\frac{1+y}{ny} < \frac{w_1}{w_N}$.
Since $f(0) = 0$, every agent gets exactly one item according to the picking sequence.
Hence, for every $B\subseteq M\setminus A_1$ with $|B|\le 1$, it holds that
\[
u_1(A_1) + y\cdot u_1(B) 
\le 1 + y 
< \frac{w_1}{w_N}\cdot ny = \frac{w_1}{w_N}\cdot (n-nx) \le \frac{w_1}{w_N}\cdot (u_1(M) - n\cdot x\cdot u_1(B)).
\]
This shows that the divisor picking sequence fails WPROP$(x,y)$, a contradiction.
Hence, $f(0) > 0$.
Now, the same argument in the proof of \eqref{eq:fa-fb} holds for every $b\geq 0$.
Applying \eqref{eq:fa-fb} for both $(0,1)$ and $(1,0)$ implies that $\frac{f(0)}{f(1)} = \frac{y}{1+y}$, and so $f(0) = y$.
\end{proof}

\subsection{Weighted Nash welfare}
\label{sec:WNW}

The WEF and WPROP criteria 
consider individual agents or pairs of agents.
A different approach, known in economics as ``welfarism'' \citep{Moulin03}, takes a global view and tries to find an allocation that maximizes a certain aggregate function of the utilities. 
A common aggregate function is the product of utilities, also called the \emph{Nash welfare}
\citep{nash1950bargaining}. This notion extends to the weighted setting as follows \cite{kalai1977nonsymmetric}.
\begin{definition}[MWNW]
A \emph{maximum weighted Nash welfare} allocation is an allocation that maximizes the weighted product $ \prod_{i\in N} u_i(A_i)^{w_i}$.
If the maximum weighted product is zero (that is, when all allocations give a utility of zero to one or more agents), then the MWNW rule first maximizes the number of agents who get a positive utility, and subject to that, maximizes the weighted product for these agents.\footnote{
An alternative tie-breaking rule for the case of zero weighted product is to maximize the \emph{total weight} of agents who get a positive utility, so that agents with higher weights are prioritized.
The choice of tie-breaking rule does not affect our results, as all relevant examples that we consider admit allocations with positive weighted Nash welfare.
}
\end{definition}

With equal entitlements, MWNW implies EF1 \citep{CaragiannisKuMo19}. 
However, with different entitlements, MWNW is incompatible with
WEF$(1,0)$ 
\citep{ChakrabortyIgSu20}
and with WPROP$(0,1)$ \citep{ChakrabortyScSu21}.
We generalize both of these  incompatibility results at once. 
\begin{theorem}
\label{thm:mwnw_wprop}
For each $x\in[0,1]$,
there exists an instance with identical items in which
every MWNW allocation is not WPROP$(x,1-x)$, and hence 
not WEF$(x,1-x)$.
\end{theorem}

\begin{proof}
For $x < 1$, consider an instance with $n$ identical items and $n$ agents,
with $n > w_1 > \frac{2-x}{1-x}$
(the range is nonempty when $n$ is sufficiently large),
and $w_2=\cdots=w_n =(n-w_1)/(n-1)$,
so $w_N = n$.
Any MWNW allocation must give a single item to each agent. But this violates 
the WPROP$(x,1-x)$ condition for agent~$1$,
since $\frac{1+(1-x)}{w_1} < \frac{2-x}{\frac{2-x}{1-x}} = 1-x = \frac{n-nx}{w_N}$.

For $x = 1$, consider 
the instance used in Proposition~4.1 of \citet{ChakrabortyIgSu20}, with $c=2$.
Specifically, there are $n=2$ agents with weights $w_1 = 1$ and $w_2 = k$ for some positive integer $k$, and $m = k+c+2$ identical items (with value~$1$ each).
As Chakraborty et al.~showed, when $k$ is sufficiently large, every MWNW allocation gives exactly one item to agent~$1$.
The WPROP$(1,0)$ condition for agent~$1$ is
\begin{align*}
\frac{u_1(A_1) + 0\cdot u_1(B)}{w_1}
& \ge 
\frac{u_1(M) - n\cdot 1\cdot u_1(B)}{w_N},
\end{align*}
where $u_1(A_1)=u_1(B)=1$ since all items are worth $1$. Therefore, the condition becomes
\begin{align*}
\frac{1}{1}
& \ge 
\frac{(k+c+2) - 2}{k+1},
\end{align*}
which is false since $c=2$. 
Therefore, every MWNW allocation is not WPROP$(1,0)$.
\end{proof}

\begin{corollary}
\label{cor:WWEF1-WPROP}
For each $x\in[0,1]$, WWEF1 does not imply WPROP$(x,1-x)$.
\end{corollary}

\begin{proof}
\citet{ChakrabortyIgSu20} proved that MWNW implies WWEF1.
However, Theorem~\ref{thm:mwnw_wprop} shows that MWNW does not imply WPROP$(x,1-x)$ for any $x\in[0,1]$.
\end{proof}

\section{Share-Based Notions}
\label{sec:maximin}
In this section, we turn our attention to share-based notions, which assign a threshold to each agent representing the agent's ``fair share''.

\subsection{Definitions}
\label{sec:share-definitions}

Denote by $\Pi(M,n)$ the collection of all ordered partitions of $M$ into $n$ subsets.
In the equal-entitlement setting, the \emph{($1$-out-of-$n$) maximin share} of an agent~$i$ is defined as follows:
\begin{align*}
\operatorname{MMS}_{i}^{1\text{-out-of-}n}(M) := \max_{(Z_1,\dots,Z_n) \in \Pi(M,n)} ~ \min_{j\in [n]} ~~~ u_i(Z_j).
\end{align*}

Sometimes we drop $M$ and the superscript `$1$-out-of-$n$' and simply write  $\operatorname{MMS}_{i}$.
An allocation is called \emph{MMS-fair} or simply \emph{MMS} if the utility that each agent receives is at least as high as the agent's MMS.
Similarly, for a parameter $\alpha$, an allocation is called \emph{$\alpha$-MMS-fair} or \emph{$\alpha$-MMS} if every agent receives utility at least $\alpha$ times her MMS.
We will use analogous terminology for other share-based notions.

There are several ways to extend this notion to the unequal-entitlement setting.
The first definition is due to \citet{FarhadiGhHa19}.
Denote by $\mathbf{w} = (w_1,\dots,w_n)$ the list of weights.

\begin{definition}[WMMS]
\label{def:WMMS}
The \emph{weighted maximin share} of an agent $i\in N$ is defined as:
\begin{align*}
\operatorname{WMMS}_{i}^{\mathbf{w}}(M) :=  \max_{(Z_1,\dots,Z_n) \in \Pi(M,n)}   \min_{j\in [n]} \frac{w_i}{w_j}\cdot u_i(Z_j)
\\
= w_i\cdot \max_{(Z_1,\dots,Z_n) \in \Pi(M,n)}   \min_{j\in [n]} \frac{u_i(Z_j)}{w_j}.
\end{align*}
\end{definition}
Sometimes we will drop $\mathbf{w}$ and $M$ from the notation when these are clear from the context; the same convention applies to other notions in this section.

Intuitively, WMMS tries to find the ``most proportional'' allocation with respect to all agents' weights and agent $i$'s utility function.
Note that the WMMS of agent~$i$ depends not only on $i$'s entitlement, but also on the entitlements of all other agents.
This might lead to unintuitive outcomes.
For example, 
suppose that in an inheritance division problem, there are three agents with entitlements $w_1=20, w_2=30, w_3=50$. Then, the court decides to move some entitlement of agent~$2$ to agent~$3$ so that the new entitlements are $w_1=20, w_2=25, w_3=55$.
Even though agent~$1$'s entitlement remains fixed, her WMMS may vary due to changes unrelated to her;\footnote{Concretely, suppose that the instance consists of three items for which agent~$1$ has utilities $20,30,50$, respectively.
Originally, agent~$1$'s WMMS is $1$.
However, after the entitlements of agents~$2$ and $3$ change, agent~$1$'s WMMS decreases to $10/11$.} this can be seen as a disadvantage of the WMMS notion.
\citet{FarhadiGhHa19} showed that a $1/n$-WMMS allocation always exists, and this guarantee is tight for every $n$.

The second definition was implicitly considered by \citet{BabaioffNiTa21}, who did not give it a name.
Following \citet{Segalhalevi19}, we call it the \emph{ordinal maximin share}.
\citet{BabaioffEzFe21} called it the \emph{pessimistic share}.
To define this share, we first extend the notion of $1$-out-of-$n$ MMS as follows. 
For all positive integers $\ell \le d$,

\begin{align*}
\operatorname{MMS}_{i}^{\ell\text{-out-of-}d}(M) := \max_{P \in \Pi(M,d)}   \min_{Z\in \operatorname{Unions}(P,\ell)} u_i(Z),
\end{align*}
where the minimum is taken over all unions of $\ell$ bundles from a given $d$-partition $P$.
Based on this generalized MMS notion, we define the ordinal MMS.

\begin{definition}[OMMS]
The \emph{ordinal maximin share} of an agent $i\in N$ is defined as:
\begin{align*}
\operatorname{OMMS}_{i}^{\mathbf{w}}(M) :=  \max_{\ell,d:~\frac{\ell}{d} \leq \frac{w_i}{w_N}} \operatorname{MMS}^{\ell\text{-out-of-}d}_i(M).
\end{align*}
\end{definition}
With equal entitlements, the OMMS reduces to the MMS \citep{Segalhalevi19}. 
Therefore, an OMMS allocation always exists for agents with identical valuations, but may not exist for $n\geq 3$ agents with different valuations \citep{KurokawaPrWa18}.
With different entitlements, it is an open question whether an OMMS allocation always exists for agents with identical valuations.

The third notion, the \emph{AnyPrice Share}, is due to \citet{BabaioffEzFe21}.
Instead of partitioning the items into $n$ disjoint bundles, an agent is allowed to choose any collection of (possibly overlapping) bundles.
However, the agent must then assign a weight to each chosen bundle so that the sum of the bundles' weights is
$w_N$, and each item belongs to bundles whose total weight is at most the agent's entitlement. 
The AnyPrice Share is the agent's utility for the least valuable chosen bundle. 

\begin{definition}[APS]
The \emph{AnyPrice share} of an agent $i\in N$ is defined as:
\begin{align*}
\operatorname{APS}_{i}^{\mathbf{w}}(M) := \max_{P \in \operatorname{AllowedBundleCollections}(M,w_i)}   \min_{Z\in P}  u_i(Z),
\end{align*}
where the maximum is taken over all collections $P$ of bundles such that for some assignment of weights to the bundles in $P$, the total weight of all bundles in $P$ is $w_N$, and for each item, the total weight of the bundles to which the item belongs is at most $w_i$.
\end{definition}

Observe that when all entitlements are equal (to $w_N/n$), the agent can choose any $1$-out-of-$n$ MMS partition and assign a weight of $w_N/n$ to each part; this shows that the APS is at least the MMS in the equal-entitlement setting. 
\citet{BabaioffEzFe21} proved that a $3/5$-APS allocation exists for agents with arbitrary entitlements, and that the APS is always at least as large as the OMMS.
Hence, their result implies the existence of a $3/5$-OMMS allocation.
Babaioff et al.~also gave an equal-entitlement example in which the APS is strictly larger than the MMS. 
Their example exhibits a disadvantage of the APS: an allocation that gives every agent her APS may not exist even when all agents have identical valuations and equal entitlements.

The fourth notion, which is new to this paper, is the \emph{normalized maximin share}.
The idea is that we take an agent's $1$-out-of-$n$ MMS and scale it according to the agent's entitlement.

\begin{definition}[NMMS]
The \emph{normalized maximin share} of an agent $i\in N$ is defined as:
\begin{align*}
\operatorname{NMMS}_{i}^{\mathbf{w}}(M) := 
\frac{w_i}{w_N}\cdot n\cdot 
\operatorname{MMS}^{1\text{-out-of-}n}_{i}(M).
\end{align*}
\end{definition}

Compared to the previous three notions, the definition of NMMS is relatively simple, which makes it easier to explain to laypeople and therefore more likely to be adopted in practice---indeed, as \citet{Procaccia19} pointed out, the explainability of fairness guarantees is crucial for convincing users in practical applications that the proposed solution is fair.
Moreover, the NMMS of an agent depends only on the agent's relative entitlement (i.e., $w_i/w_N$); this is in contrast to the WMMS, which can change even when the agent's relative entitlement stays the same (see the discussion after \Cref{def:WMMS}).

We illustrate these notions with a small example.

\begin{example}
Let $n=m=3$, $w_1=1$, $w_2=2$, $w_3=4$, and suppose that agent~$1$ has utility $1$ for every item.
\begin{itemize}
\item For WMMS, each set $Z_j$ in a maximizing partition must contain exactly one item; otherwise at least one of them is empty and the minimum over $j\in[n]$ is $0$. This means that $\operatorname{WMMS}_{1}^{\mathbf{w}}(M) = w_1\cdot\min_{j\in[3]}\frac{1}{w_j} = \frac{1}{4}$.
\item For OMMS, note that $\operatorname{MMS}^{\ell\text{-out-of-}d}_i(M) = 0$ whenever $\ell\le d-3$, because for every partition of $M$ into $d$ sets, at least $d-3$ sets are empty.
Since $\ell \le d-3$ always holds when $\ell/d \le w_i/w_N = 1/7$, we have $\operatorname{OMMS}_{1}^{\mathbf{w}}(M) = 0$.
\item For APS, note that if we distribute a weight of $7$ among nonempty bundles, some item necessarily receives weight strictly larger than $1$.
Indeed, this is true if some bundle receives weight larger than $1$.
Otherwise, each nonempty bundle must receive weight exactly $1$ (since there are $2^3 - 1 = 7$ nonempty bundles in total), and this is again true.
Hence, a positive weight must be assigned to the empty bundle, and $\operatorname{APS}_{1}^{\mathbf{w}}(M) = 0$.
\item For NMMS, similarly to WMMS, each set $Z_j$ in a maximizing partition must contain exactly one item.
Thus, we have $\operatorname{NMMS}_{1}^{\mathbf{w}}(M) = \frac{w_1}{w_N}\cdot n\cdot \operatorname{MMS}^{1\text{-out-of-}3}_{1}(M) = \frac{3}{7}$.
\end{itemize}
\end{example}

For example scenarios in which some of these share-based notions better reflect intuitive perceptions of fairness than others, we refer to the instances in the proof of \Cref{thm:aps_mms}.

\subsection{Ordinal versus cardinal notions}

APS-fairness and OMMS-fairness are both ``ordinal'' notions in the sense that, even though the APS and OMMS values are numerical, whether an allocation is APS- or OMMS-fair can be determined by only inspecting each agent's ordinal ranking over bundles.
On the other hand, WMMS-fairness and NMMS-fairness are both ``cardinal'' notions, since they depend crucially on the numerical utilities.
We show that there are no implication relations between  ordinal and cardinal notions: each type of notion does not imply any nontrivial approximation of the other type.

\begin{theorem}
\label{thm:aps_mms}
(a)
An APS-fair or OMMS-fair allocation does not necessarily yield any positive approximation of WMMS-fairness or NMMS-fairness.

(b) 
A WMMS-fair or NMMS-fair allocation does not necessarily yield any positive approximation of APS-fairness or OMMS-fairness.
\end{theorem}

\begin{proof}
For both parts we consider agents who share the same utility function. 
Since an agent's APS is always at least as large as her OMMS \citep{BabaioffEzFe21}, it suffices to consider APS for part (a) and OMMS for part (b). 

(a) 
There are two items with utilities $40$ and $60$, and two agents with weights $w_1 = 0.4$ and $w_2 = 0.6$.
Agent~$1$'s APS is $0$, since for every assignment of weights to bundles, each item must have a total weight of at most $0.4$, so the empty bundle must have a weight of at least $0.2$. 
Hence, the allocation giving both items to agent~$2$ is APS-fair; this leaves agent~$1$ with utility $0$.
On the other hand, 
agent~$1$'s WMMS is $40$.
Moreover, the MMS is $40$, so agent~$1$'s NMMS is $0.4\cdot2\cdot 40 = 32$. 
Therefore, the above APS-fair allocation does not yield any positive approximation of the NMMS or WMMS for agent~$1$.

(b) Consider the same items as above, but now there are three agents with weights $w_1=w_2=0.2$ and $w_3=0.6$.
The WMMS and NMMS are always $0$ when $n>m$, so every allocation is WMMS-fair and NMMS-fair.
However, taking $(\ell,d) = (1,2)$ in the definition of OMMS, we find that agent~$3$'s OMMS is at least $40$.
Therefore, a WMMS-fair and NMMS-fair allocation guarantees no positive fraction of the OMMS.
\end{proof}

\subsection{Approximation of NMMS}

Given that NMMS is a new notion, an important question is whether any useful approximation of it can be ensured.
\citet{FarhadiGhHa19} proved that the best possible WMMS guarantee is $1/n$-WMMS. 
We show that the same holds for NMMS, starting with the upper bound.

\begin{theorem}
\label{thm:nmms_best}
For every integer $n\geq 2$, and every real number $r>1/n$, there is an instance with $n$ agents in which no allocation is $r$-NMMS.
\end{theorem}

\begin{proof}
We use the same instance as in Theorem~2.1 of \citet{FarhadiGhHa19}.
There are $n$ agents and $2n-1$ items.
Each agent $i < n$ has weight $w_i=\varepsilon\in (0, 1/n)$ while agent~$n$ has weight $w_n = 1-(n-1)\varepsilon$. The utilities are as follows:
\begin{align*}
    u_i(j)&=\begin{cases}
    \varepsilon &\text{if } j < n,\\
    1-(n-1)\varepsilon &\text{if } j =n,\\
    0 &\text{otherwise},
    \end{cases} \quad \forall i <n; \\
    u_n(j)&= \begin{cases}
    \frac{1-(n-1)\varepsilon}{n} &\text{if } j \le n,\\
    \varepsilon &\text{if } j > n.
    \end{cases}
\end{align*}
The MMS, and therefore the NMMS, of each of the first $n-1$ agents is positive, as they all value the first $n$ items positively.
Thus, in order to provide them with a positive approximation of their NMMS, we must assign to each of them at least one of the first $n$ items. 
This means that agent~$n$ is left with at most one of the first $n$ items, and therefore utility at most 
$$\beta := \frac{1-(n-1)\varepsilon}{n}+(n-1)\varepsilon=\frac{1+(n-1)^2 \varepsilon}{n}.$$ 
On the other hand, for sufficiently small $\varepsilon$, agent~$n$'s MMS is $\frac{1-(n-1)\varepsilon}{n}$, so her NMMS is $(1-(n-1)\varepsilon)^2$.
Therefore, the approximation ratio of every allocation is at most
\[
\frac{1+(n-1)^2\varepsilon}{n\cdot (1-(n-1)\varepsilon)^2}.
\]
When $\varepsilon \rightarrow 0$, this expression approaches $1/n$, which proves the claim.
\end{proof}

We establish a matching lower bound by showing that \wpropstar{(1,0)}, and therefore WEF$(1,0)$, implies \nnmms.

\begin{lemma}
\label{lem:strongwprop10_nmms}
\wpropstar{(1,0)} implies \nnmms.
\end{lemma}

\begin{proof}
If there are fewer than $n$ items, every agent's NMMS is $0$, and the statement holds trivially.
Assume therefore that there are at least $n$ items.
Consider an arbitrary agent $i \in N$, and let $B^*_i$ be the set of the $n-1$ most valuable items for $i$, ties broken arbitrarily.
In any partition of $M$ into $n$ bundles, at least one bundle does not contain any item from the set $B^*_i$, and thus it is contained in $M\setminus B^*_i$.
This means that
\begin{align}
    \operatorname{MMS}_i \le u_i(M \setminus B^*_{i}). \label{ineq:mms_minusbestnm1}
\end{align}
In every \wpropstar{(1,0)} allocation $A$, there exists $B_j\subseteq A_j$ of size at most $1$ for each $j\in N\setminus \{i\}$ such that
\begin{align*}
u_i(A_i) 
&\ge \frac{w_i}{w_N} \left( u_i(M) - \sum_{j \in N \setminus \{i\}} u_i(B_j)\right) \\
&\ge \frac{w_i}{w_N} \left( u_i(M) -  u_i\left(B_i^*\right)\right) \\
&\ge \frac{w_i}{w_N} \cdot \operatorname{MMS}_i && \text{(by \eqref{ineq:mms_minusbestnm1})}
\\
&= \frac{1}{n}\cdot \operatorname{NMMS}_i,
\end{align*}
so the allocation $A$ is \nnmms, as desired.
\end{proof}

\begin{theorem}
\label{thm:wef10_nmms}
WEF$(1,0)$ implies \nnmms. 
In particular, every instance admits a \nnmms allocation.
The factor $1/n$ in both statements cannot be improved.
\end{theorem}

\Cref{thm:wef10_nmms} is an immediate consequence of \Cref{lem:WEF-implies-strongWPROP,lem:strongwprop10_nmms} and \Cref{thm:nmms_best}.
It generalizes the result that EF1 implies $1/n$-MMS in the unweighted setting \citep[Prop.~3.6]{AmanatidisBiMa18}, and stands in contrast to the result that WEF$(1,0)$ does not imply any positive approximation of WMMS \citep[Prop.~6.2]{ChakrabortyIgSu20}.
Since 
a weighted round-robin algorithm as well as a generalization of \citet{BarmanKrVa18}'s market-based algorithm ensure WEF$(1,0)$
\citep{ChakrabortyIgSu20}, these algorithms guarantee \nnmms as well.

The next theorem shows that, if each agent values each item at most her NMMS, then the approximation factor can be improved to $1/2$, thereby providing a direct analog to Theorem~3.2 of \citet{FarhadiGhHa19} for WMMS.

\begin{theorem}
\label{thm:nmms-half}
Any instance in which $u_i(g) \leq \emph{NMMS}_i$ for all $i\in N$ and $g\in M$
admits a $1/2$-NMMS allocation.
\end{theorem}

\begin{proof}
We use an algorithm similar to Algorithm~2 of \citet{FarhadiGhHa19} for WMMS.
We normalize the weights so that $w_N = 1$.
The algorithm proceeds as follows:
\begin{enumerate}
\item Scale the utilities of each agent $i$ so that NMMS$_i=w_i n$.
Consequently, we have $\text{MMS}^{1\text{-out-of-}n}_{i}(M) = 1$, and therefore $u_i(M)\geq n$.

\item \label{loop} Let the set of \emph{unsatisfied agents} be the set of agents $i$ whose current utility is less than $w_i n/2$.
If this set is empty, the algorithm terminates.
\item Pick a pair $(i,g)\in N\times M$ that maximizes $u_i(g)$ among all unsatisfied agents and unallocated items, breaking ties arbitrarily. Allocate $g$ to $i$.
\item If there are remaining items, go back to Step~\ref{loop}. 
Else, the algorithm terminates.
\end{enumerate}

We will show that the allocation returned by the algorithm is $1/2$-NMMS.
Note that this allocation may be incomplete, in which case we can allocate the remaining items arbitrarily in a post-processing step; this does not affect the $1/2$-NMMS guarantee.

We first prove the following auxiliary claim on the (possibly incomplete) allocation returned by the algorithm:
\emph{For each agent $i$, $u_i(A_i)\leq w_i n$}.
By the assumption of the theorem and the normalization in Step~1 of the algorithm, the claim holds for each agent~$i$ who receives at most one item.
If an agent has already received one or more items and is still unsatisfied, this means that all items allocated to her together yield utility less than $w_i n / 2$. 
Since items are allocated to each agent in descending order of their utility to the agent, it follows that every remaining item yields utility less than $w_i n / 2$ to $i$. As a result, $i$'s utility throughout the algorithm is always smaller than $w_i n / 2 + w_i n / 2 = w_i n$, thereby proving the claim.

Now, since each item is given to an unsatisfied agent who values it the most, the auxiliary claim implies that $u_i(A_j)\leq u_j(A_j) \le w_j n$ for all agents $i,j$ as long as $i$ is unsatisfied.
We will show that all agents are satisfied when the algorithm terminates.
Assume for contradiction that agent~$i$ remains unsatisfied at the end of the algorithm.
Summing the previous inequality over all $j\neq i$ 
yields that, from $i$'s perspective, the total value of items allocated to all other agents is at most $\sum_{j\neq i}w_j n = n - w_i n$.
Since $u_i(M)\geq n$, the total value of remaining items for $i$ is at least $w_i n$, so agent $i$ should be satisfied, a contradiction.
\end{proof}

We know from \Cref{thm:wef10_nmms} that WEF$(1,0)$ implies \nnmms.
The next result shows that WEF$(x,1-x)$ for other values of $x$ performs far worse; in fact, we will establish a stronger statement with WEF$(x,y)$.

\begin{theorem}
\label{thm:wefx1mx_nmms}
For all $n\ge 2$, $x\geq 0$, and $y>0$, WEF$(x,y)$ does not imply any positive approximation of NMMS,
even for $n$ agents with identical valuations.
\end{theorem}

\begin{proof}
Fix $n\geq 2$.
Assume that the agents have identical valuations and there are $m = n$ items:
items $1,\ldots,n-1$ have value $1$ and item $n$ has value $x+y$.
As there are $n$ items with positive value, the MMS, and therefore the NMMS, of every agent is positive.

Pick positive real numbers $w, w'$ such that $w'=\frac{1-w}{n-1}$ and 
\begin{align}
\label{eq:ww'}
    0 < \frac{w}{w'} < \frac{y}{1+y}.
\end{align}
Suppose that agent~$1$ has weight $w$ and every other agent has weight $w'$.
Note that $w<w'$ and $w+(n-1)w'=1$, which implies that $w<\frac{1}{n}<w'$.

Consider an allocation that assigns no item to agent~$1$, two items---one item of value $1$ along with item~$n$ ---to agent~$2$, and one item of value $1$ to each remaining agent. 
We claim that this allocation is WEF$(x,y)$.
The condition holds for agents~$2,3,\dots,n$, since they have the same weights and the difference in their utilities is at most $x+y$.
To establish the claim, it therefore suffices to verify the WEF$(x,y)$ condition for agent~$1$ towards agent~$2$, that is,
\[
\frac{0 + y}{w} \ge \frac{1+(x+y)-x}{w'},
\]
which holds by \eqref{eq:ww'}.
Hence, the allocation under consideration is WEF$(x,y)$.
However, agent $1$ receives utility $0$ in this allocation even though her NMMS is positive.
\end{proof}

Note also that in the instance from this proof, whenever $x+y\leq 1$,
agent $1$ must receive an empty bundle in every WEF$(x,y)$ allocation.
Indeed, by \eqref{eq:ww'} we have $\frac{1+y}{w'} < \frac{(x+y)-x}{w}$, so if agent $1$ receives even item $n$ (which has the lowest value), then for another agent who receives one item with value~$1$, the WEF$(x,y)$ condition towards agent $1$ is violated. 
Therefore, when $x+y\leq 1$ and $y>0$, WEF$(x,y)$ is  incompatible with any positive approximation of NMMS.

\Cref{thm:wefx1mx_nmms} implies that WPROP$(x,y)$ does not yield any positive approximation of NMMS for all fixed $x\geq 0$ and $y>0$.
We show next that the same is true for WPROP$(1,0)$, which reveals a weakness of WPROP$(1,0)$ in comparison to \wpropstar{(1,0)} (cf.~\Cref{lem:strongwprop10_nmms}).

\begin{corollary}
\label{cor:wprop_nmms}
For each $x\in[0,1]$, WPROP$(x,1-x)$ does not imply any positive approximation of NMMS. 
\end{corollary}

\begin{proof}
The case $x\in[0,1)$ follows from \Cref{thm:wefx1mx_nmms},
so assume that $x = 1$.
Suppose that all agents have the same weight and there are $m = n$ items, each of which yields value $1$ to every agent.
Consider an allocation that assigns no item to agent~$1$, two items to agent~$2$, and one item to each remaining agent.
We claim that this allocation is WPROP$(1,0)$.
Indeed, for each $i\in N$, 
we can choose the set $B\subseteq M\setminus  A_i$ to contain a single item with value $1$, so
\begin{align*}
\frac{u_i(A_i)}{w_i}
& \ge 
0
=
 \frac{n - n\cdot 1}{w_N}
 =
\frac{u_i(M) - n \cdot u_i(B)}{w_N}.
\end{align*}
(The only exception is when $n = i = 2$, in which case we have $A_i = M$ and the WPROP$(1,0)$ property also holds.)
In this allocation, agent~$1$ receives utility $0$ even though her NMMS is positive.
\end{proof}

The following is also a consequence of \Cref{thm:wefx1mx_nmms}.

\begin{corollary}
\label{cor:WWEF1-NMMS}
WWEF1 does not imply any positive approximation of NMMS.
\end{corollary}

\subsection{Maximum weighted Nash welfare and share-based notions}

We have shown in \Cref{sec:WNW} that an MWNW allocation is not guaranteed to satisfy WEF$(x,1-x)$ or WPROP$(x,1-x)$ for each $x\in[0,1]$.
Could MWNW nevertheless imply some approximation of NMMS or WMMS? 
The answer is again negative, even when $n=2$.

\begin{theorem}\label{thm:mwnw_nmms}
Even for $n=2$, MWNW does not imply any positive approximation of NMMS.
\end{theorem}

\begin{proof}
Let $f:[1,\infty)\rightarrow [0,\infty)$ be the function $f(z) = \frac{2^{z-1}-1}{1+z}$ for all $z\in[1,\infty)$.
Note that $f(1) = 0$ and $f(z)\rightarrow\infty$ as $z\rightarrow\infty$.
Moreover, we have
\[
f'(z) = \frac{2^{z-1}((1+z)\ln 2 - 1) + 1}{(1+z)^2} > 0
\]
for all $z\ge 1$, since $(1+z)\ln 2 - 1 \ge 2\ln 2 - 1 > 0$.
Hence, $f$ is strictly increasing in $[1,\infty)$, and the inverse function $f^{-1}$ is well-defined
in $[0,\infty)$.
For $\varepsilon \in (0,0.5)$, we have $\frac{1}{2\varepsilon} > 1$, 
so $f^{-1}\left(\frac{1}{2\varepsilon}\right)>f^{-1}(1) > 3.44$
and
$\frac{1}{f^{-1}\left(\frac{1}{2\varepsilon}\right)} < 0.3$.

Let $n=2$, $m=3$, $\varepsilon < 0.5$ be a small positive constant, and $0 < w < \frac{1}{f^{-1}\left(\frac{1}{2\varepsilon}\right)}$.
The weights are $w_1 = w$ and $w_2 = 1-w$.
The agents' utilities for the items are given by the following table:

\renewcommand{\arraystretch}{1.2}
\begin{center}
\begin{tabular}{|c|c|c|c|}
\hline
Item   & $1$ & $2$ & $3$ \\
\hline
$u_1$ & $2\varepsilon$ & $1+\frac{1}{w}$ & $1+\frac{1}{w}$ \\
\hline
$u_2$ & $0$ & $1$ & $1$ \\
\hline
\end{tabular}
\end{center}

Clearly, every MWNW allocation must assign item~$1$ to agent~$1$, and agent~$2$ must receive either one or both of items~$2$ and $3$, which are interchangeable.
We claim that the unique MWNW allocation assigns item~$1$ to agent~$1$ and items~$2$ and $3$ to agent~$2$.
To prove this claim, it suffices to show that 
\begin{align*}
(2\varepsilon)^w\cdot 2^{1-w} &> \left(2\varepsilon + 1 + \frac{1}{w}\right)^w\cdot 1^{1-w}.
\end{align*}
This reduces to
\[
2^{1-w} > \left(1 + \frac{1+\frac{1}{w}}{2\varepsilon}\right)^w,
\]
or equivalently,
\[
\frac{2^{\frac{1}{w} - 1}-1}{1+\frac{1}{w}} > \frac{1}{2\varepsilon}.
\]
The left-hand side is simply $f(1/w)$.
Since $f^{-1}\left(1/2\varepsilon\right) < 1/w$ and $f$ is strictly increasing, we have $1/2\varepsilon < f\left(1/w\right)$, as desired.
This proves the claim that agent~$1$ receives only item~$1$ in the unique MWNW allocation, and thereby a utility of $2\varepsilon$.

Finally, since $2\varepsilon < 1 < 1 + \frac{1}{w}$, agent~$1$'s MMS is $1+\frac{1}{w}$.
Therefore, her NMMS is $\frac{w}{w+(1-w)}\cdot 2\cdot\left(1+\frac{1}{w}\right) = 2(w+1) > 2$.
This means that in the MWNW allocation, agent~$1$ receives less than $\varepsilon$ times her NMMS.
The conclusion of the theorem follows by taking $\varepsilon\rightarrow 0$.
\end{proof}

The example in \Cref{thm:mwnw_nmms} can also be used to show an analogous result for WMMS.

\begin{theorem}
\label{thm:mwnw_wmms}
Even for $n=2$, MWNW does not imply any positive approximation of WMMS.
\end{theorem}

\begin{proof}
Consider the instance in \Cref{thm:mwnw_nmms}, where we have shown that agent~$1$ receives utility $2\varepsilon$ in the unique MWNW allocation.
By considering the partition in which agent~$1$ receives item~$3$ and agent~$2$ receives items $1$ and $2$, we find that agent~$1$'s WMMS is at least
\[
\min\left\{1+\frac{1}{w}, \frac{w}{1-w}\cdot\left(1+\frac{1}{w}+2\varepsilon\right)\right\}.
\]
Clearly, $1+\frac{1}{w} > 1$.
Moreover, $\frac{w}{1-w}\cdot\left(1+\frac{1}{w}+2\varepsilon\right) = \frac{w+1+2\varepsilon w}{1-w}$ is greater than $1$ because the numerator is greater than $1$ while the denominator is less than $1$.
Hence, agent~$1$'s WMMS is at least $1$.
The desired conclusion follows by taking $\varepsilon\rightarrow 0$.
\end{proof}

In the unweighted setting, \citet{CaragiannisKuMo19} proved that maximum Nash welfare implies $\Theta(1/\sqrt{n})$-MMS.
\Cref{thm:mwnw_nmms,thm:mwnw_wmms} show that this result cannot be generalized to the weighted setting through NMMS or WMMS.

\subsection{Ordered-EF1 allocations}
\label{sec:OEF1}

\citet{FarhadiGhHa19} established the $1/n$-WMMS guarantee through a well-known algorithm, the \emph{(unweighted) round-robin algorithm}, which simply lets agents take turns picking their favorite item from the remaining items.
A crucial specification required for their guarantee to work is that the agents must take turns in non-increasing order of their weights.
This motivates the following definition.

\begin{definition}[OEF1]\label{def:orderedef1}
An allocation is \textit{ordered-EF1} if 
\begin{enumerate}[label=(\roman*)]
\item it is EF1 when we disregard weights, and
\item the agents can be renumbered so that $w_1 \ge w_2 \ge \dots \ge w_n$ and no agent $i \in N$ has (unweighted) envy towards any later agent $j \in \{i+1, i+2, \dots, n\}$.
\end{enumerate}
\end{definition}

The standard proof that the unweighted round-robin algorithm outputs an EF1 allocation \citep[p.~7]{CaragiannisKuMo19} implies that the algorithm with the aforementioned ordering specification outputs an OEF1 allocation.
Moreover, it follows from \citet{FarhadiGhHa19}'s proof that OEF1 implies $1/n$-WMMS.
We show next that, interestingly, OEF1 also implies \nnmms, which means that the same algorithm guarantees the optimal (worst-case) approximation of both WMMS and NMMS.

\begin{theorem}
\label{thm:oef1_nmms}
Any OEF1 allocation is also \nnmms.
\end{theorem}

\begin{proof}
If $m< n$, every agent's NMMS is $0$ and the theorem holds trivially, so assume that $m \ge n$.
Consider an OEF1 allocation $A$, where the agents have been renumbered as in \Cref{def:orderedef1}.
We prove the $1/n$-NMMS condition for an arbitrary agent $i\in N$.
For each $j\in\{i+1,\dots,n\}$, we have
\begin{align*}
    u_i(A_i) &\ge u_i(A_j),
\end{align*}
while for each $j \in \{1,\dots,i-1\}$, there exists $B_j\subseteq A_j$ of size at most $1$ such that
\begin{align*}
    u_i(A_i) &\ge u_i(A_j) - u_i(B_j).
\end{align*}
Adding the above inequalities for all $j\in N\setminus\{i\}$ and then adding $u_i(A_i)$ to both sides, we get
\begin{align*}
    n \cdot u_i(A_i) \ge u_i(M) - \sum_{j=1}^{i-1} u_i(B_j).
\end{align*}
Let $H_{i-1}$ be a set containing the $i-1$ most valuable items in $M$ according to $u_i$, ties broken arbitrarily.
We have $\sum_{j=1}^{i-1}u_i(B_j) \le u_i(H_{i-1})$, so the above inequality implies
\begin{align*}
    u_i(A_i) \ge \frac{u_i(M\setminus H_{i-1})}{n}.
\end{align*}
Thus, in order to establish \nnmms for agent~$i$, it suffices to show that
\begin{align*}
u_i(M\setminus H_{i-1}) \ge n\cdot \frac{w_i}{w_N} \cdot \operatorname{MMS}_i.
\end{align*}

To this end, consider agent~$i$'s maximin partition, i.e., a partition $P^* = \{P_1^*,P_2^*,\dots,P_n^*\}$ such that $u_i(P_j^*) \ge \operatorname{MMS}_i$ for all $j\in[n]$.
Since $H_{i-1}$ contains $i-1$ items, at least $n-i+1$ bundles in $P^*$ do not contain any item from $H_{i-1}$.
This means that $u_i(M\setminus H_{i-1})\ge (n-i+1)\cdot \operatorname{MMS}_i$.
It therefore remains to show that
\begin{align*}
\frac{n-i+1}{n} \ge \frac{w_i}{w_N}.
\end{align*}
Recall that $w_i$ is the $i$-th largest weight, so $w_i/w_N \le 1/i$, and the above inequality reduces to showing that $\frac{n-i+1}{n} \ge \frac{1}{i}$.
This is equivalent to $(n-i)(i-1) \ge 0$, which indeed holds since $1\le i\le n$, completing the proof.
\end{proof}

\subsection{Other relationships between fairness notions}

WMMS and NMMS are both cardinal extensions of MMS to the weighted setting.
However, we show in the next two theorems that the relationship between these two notions is rather weak.

\begin{theorem}
\label{thm:wmms_nmms}
WMMS implies \nnmms, and the factor $1/n$ is tight.
\end{theorem}

Note that the tightness does not follow from \Cref{thm:nmms_best}, since a WMMS allocation does not always exist.

\begin{proof}
By definition, if $A$ is a WMMS allocation, then for each agent $i \in N$,
\begin{align*}
u_i(A_i) 
&\ge w_i\cdot \max_{(Z_1,\dots,Z_n) \in \Pi(M,n)}   \min_{j\in [n]} \frac{u_i(Z_j)}{w_j} \\
&\ge w_i\cdot \max_{(Z_1,\dots,Z_n) \in \Pi(M,n)}   \min_{j\in [n]} \frac{u_i(Z_j)}{w_N} \\
&= \frac{w_i}{w_N} \cdot \max_{(Z_1,\dots,Z_n) \in \Pi(M,n)}   \min_{j\in [n]} u_i(Z_j) = \frac{w_i}{w_N} \cdot \operatorname{MMS}_i = \frac{1}{n}\cdot \operatorname{NMMS}_i.
\end{align*}
This establishes the first part of the theorem.

For the tightness, suppose there are $m=n$ items.
Agent~$1$ has utility~$1$ for every item, while for each $i\ge 2$, agent~$i$ has utility~$1$ for item $i$ and $0$ for all other items.
Suppose that $w_1 = 1-(n-1)\varepsilon$ and $w_2 = \dots = w_n = \varepsilon$ for some small $\varepsilon > 0$.
Consider the allocation that assigns item~$i$ to agent~$i$ for all $i\in N$.
This is clearly WMMS-fair for agents $2,\dots,n$.
Agent~$1$'s WMMS is $1$, and she also receives utility $1$, so the allocation is WMMS-fair for her as well.
On the other hand, her NMMS is $w_1\cdot n$, which converges to $n$ as $\varepsilon\rightarrow 0$.
Hence, WMMS cannot offer more than a \nnmms guarantee.
\end{proof}

\begin{theorem}
\label{thm:nmms_wmms}
For $n\ge 3$, NMMS does not imply any positive approximation of WMMS.
\end{theorem}

\begin{proof}
Let $n\ge 3$, and consider an instance with $n$ agents of weights $w_1 > w_2 > w_3 = \dots = w_n$ such that $w_1 + \dots + w_n = 1$, and $n$ items.
Agent~$1$ has utility $w_2$ for item~$1$, $w_1$ for item~$2$, and $w_j$ for item $j\in\{3,\dots,n\}$.
For each $i\ge 2$, agent~$i$ has utility $1$ for item $i$ and $0$ for all other items.
Consider the allocation that assigns item~$i$ to agent~$i$ for all $i\in N$.
This is clearly NMMS-fair for agents $2,\dots,n$.
Agent~$1$'s WMMS is $w_1$, while her NMMS is $nw_1w_n$.

We choose $w_1,\dots,w_n$ in such a way that $w_1$ is close to $1$ and $w_2$ is sufficiently close to $1-w_1$ so that $w_2 > n w_1 w_n$ (note that once $w_1$ and $w_2$ are fixed, $w_3,\dots,w_n$ are also fixed).
Since agent~$1$ receives utility $w_2$ from the allocation above, this allocation is NMMS-fair.
However, as $w_1\rightarrow 1$, the ratio $w_2/w_1$ converges to $0$, so NMMS-fairness cannot imply any positive approximation of WMMS-fairness.
\end{proof}

\citet{ChakrabortyIgSu20} showed that WEF$(1,0)$ does not guarantee any positive approximation of WMMS.
We now prove a similar negative result for WEF$(x,1-x)$ and WPROP$(x,1-x)$ for all $x\in[0,1]$.
This mirrors the negative result for NMMS (\Cref{thm:wefx1mx_nmms}), with the difference being that WEF$(1,0)$ implies $1/n$-NMMS (\Cref{thm:wef10_nmms}).

\begin{theorem}
\label{thm:wef_wmms}
For each $x\in[0,1]$, WEF$(x,1-x)$ does not imply any positive approximation of WMMS.
\end{theorem}

\begin{proof}
The case $x = 1$ was already handled by Proposition~6.2 of \citet{ChakrabortyIgSu20}, so assume that $x\in [0,1)$. 

Suppose we have $n$ agents with weights satisfying $1 > w_1 > w_2 > \cdots > w_{n-1} > \frac{w_n}{2} \left(1+\sqrt{\frac{5-x}{1-x}} \right)  >0$, and $m = n$ items. 
Note that since $\sqrt{\frac{5-x}{1-x}} > 1$, we have $w_i > w_n$ for all $i\in N\setminus\{n\}$.
For each $i \in N \setminus \{n\}$, let agent $i$ value item $i$ at $1$ and every other item at $0$; let agent $n$ value item $j$ at $w_j$ for all $j\in[n]$. 
Hence, agent~$n$'s WMMS is $w_n > 0$.

Now, consider the allocation that assigns item $i$ to agent $i$ for every $i \in N \setminus \{n-1,n\}$, items $n-1$ and $n$ to agent $n-1$, and no item to agent $n$. 
We claim that this allocation is WEF$(x,1-x)$.
To prove this claim, it suffices to show that the WEF$(x,1-x)$ condition from agent~$n$ towards
all other agents is fulfilled.
The condition towards agent~$j \le n-2$ is
\begin{align*}
\frac{0+(1-x)\cdot w_j}{w_n} \ge \frac{w_j - x \cdot w_j}{w_j},
\end{align*}
which holds since $w_j > w_n$.
The condition from agent~$n$ towards agent~$n-1$ is
\begin{align*}
\frac{0+(1-x)\cdot w_{n-1}}{w_n} \ge \frac{(w_{n-1}+w_n)-x\cdot w_{n-1}}{w_{n-1}}.
\end{align*}
Upon simplification, this becomes
\begin{align*}
(1-x)\gamma_n^2 - (1-x) \gamma_n -1 \ge 0,
\end{align*}
where $\gamma_n := \frac{w_{n-1}}{w_n}$. 
Note that the roots of the equation 
\[
(1-x)t^2 - (1-x)t - 1 = 0
\]
are $t = \frac{1}{2}\left(1\pm\sqrt{\frac{5-x}{1-x}} \right)$.
Since $\gamma_n > \frac{1}{2}\left(1+\sqrt{\frac{5-x}{1-x}} \right)$, the desired inequality holds
and the allocation  is
WEF$(x,1-x)$.
But agent~$n$ receives utility $0$ in this allocation even though her WMMS is positive.
\end{proof}

\begin{corollary}
\label{cor:wpropx1mx_noapprox_wmms}
For each $x\in [0,1]$, WPROP$(x,1-x)$ does not imply any positive approximation of WMMS.
\end{corollary}

\section{Identical Items: Lower and Upper Quota}
\label{sec:identical}

Thus far, we have seen many fairness notions and proved that several of them are incompatible with, or do not imply any approximation of, one another. 
In this section, we add another dimension to the comparison between fairness notions by focusing on the case where all items are identical.
The motivation for studying this case is twofold.
First, due to its simplicity, it is easier to agree on a fairness criterion. 
If the items were divisible, agent~$i$ should clearly receive her \emph{quota} of $q_i := \frac{w_i}{w_N}\cdot m$ items.
With indivisible items, one may therefore expect agent~$i$ to receive either her \emph{lower quota} $\lfloor q_i\rfloor$ or her \emph{upper quota} $\lceil q_i\rceil$.
Second, the case of identical items is practically relevant when allocating parliament seats among states or parties, a setting commonly known as  \emph{apportionment} \citep{BalinskiYo75,BalinskiYo01}. Both quotas are frequently considered in apportionment.

Throughout this section,
we denote an allocation by $(a_1,\dots,a_n)$, where $a_i = |A_i| = $ the number of items that agent~$i$ receives.
Assume without loss of generality that every agent has utility~$1$ for every item.
We determine whether each fairness notion implies lower quota (resp., upper quota) in the identical-item setting, i.e., whether it ensures that each agent $i\in N$ receives at least $\lfloor q_i\rfloor$ items (resp., at most $\lceil q_i\rceil$ items).
Our results are summarized in Table~\ref{table:quota}.

\renewcommand{\arraystretch}{1.2}
\begin{table}[!ht]
\centering
\begin{tabular}{| c | c | c |}
\hline
\textbf{Notion} & \textbf{Lower quota} & \textbf{Upper quota} \\ \hline \hline
WEF$(x,1-x)$ & \begin{tabular}[c]{@{}c@{}}
Yes if $x = 0$
\\No if $x\in (0,1]$
\end{tabular}  & \begin{tabular}[c]{@{}c@{}}
Yes if $x = 1$
\\No if $x\in [0,1)$
\end{tabular} \\ \hline
WPROP$(x,1-x)$ & No & No \\ \hline
MWNW & No & No \\ \hline
WMMS & No & No \\ \hline
NMMS & No & No \\ \hline
OMMS & Yes & No \\ \hline
APS & Yes & No \\ \hline
\end{tabular}
\vspace{2mm}
\caption{Summary of whether each fairness notion implies lower or upper quota in the identical-item setting.
For $n = 2$, MWNW as well as WMMS-, NMMS-, OMMS-, and APS-fairness imply both lower and upper quota.
MWNW satisfies upper quota for $n = 3$.
OMMS- and APS-fairness are equivalent to lower quota for all $n$.
}
\label{table:quota}
\end{table}

\begin{example}
\label{ex:quotas}
To illustrate the ideas behind some of these proofs,
consider an instance with $n=m=3$ and entitlements $4,~1,~1$.
The quotas are $2,~ 0.5,~ 0.5$.
The unique allocation satisfying WEF$(1,0)$ 
is $(1,1,1)$:
any agent that gets no item would have weighted envy towards any agent that gets two or more items, even after removing a single item.
Note that $(1,1,1)$ is also the unique MWNW allocation as well as the unique allocation satisfying WMMS- and NMMS-fairness, since the WMMS and NMMS are positive for all agents. 
However, this allocation violates the lower quota of agent~$1$, which is $2$ items.

In contrast, the unique allocation satisfying WEF$(0,1)$ is $(3,0,0)$:
if agent~$1$ gets two or fewer items, she feels weighted envy---even after getting an additional item---towards another agent who gets (at least) one item.
The allocation $(3,0,0)$ also satisfies OMMS- and APS-fairness, as the OMMS and APS of agents~$2$ and $3$ are both $0$. 
However, it violates the upper quota of agent~$1$. 
\end{example}

We now proceed to the proofs.

\begin{theorem}
\label{thm:wef_quota}
With identical items, WEF$(x,1-x)$ implies upper quota if and only if $x = 1$, and implies lower quota if and only if $x=0$.
\end{theorem}

\begin{proof}
We start with the positive claims.

First, let $x = 1$,
and consider any WEF$(1,0)$ allocation.
For any $i,j\in N$, it holds that $\frac{a_i}{w_i} \ge \frac{a_j-1}{w_j}$, that is,
$a_i\geq w_i \cdot \frac{a_j-1}{w_j}$.
Fixing $j$ and summing this over all $i\in N\setminus\{j\}$, we get
\[
m - a_j \geq (w_N-w_j)\cdot \frac{a_j-1}{w_j}.
\]
Adding $a_j-1$ to both sides yields
\[
m - 1 \geq w_N\cdot \frac{a_j-1}{w_j},
\]
which implies that $a_j \le q_j - \frac{w_j}{w_N} + 1 < \ceil{q_j}+1$.
Since $a_j$ and $\ceil{q_j}$ are integers, we have $a_j \le \ceil{q_j}$, meaning that upper quota is fulfilled.

Next, let $x = 0$, and consider any WEF$(0,1)$ allocation.
For any $i,j\in N$, it holds that $\frac{a_i+1}{w_i} \ge \frac{a_j}{w_j}$, that is,
$w_j\cdot \left(\frac{a_i+1}{w_i}\right) \ge a_j$.
Fixing $i$ and summing this over all $j\in N\setminus\{i\}$, we get
\[
(w_N - w_i)\left(\frac{a_i+1}{w_i}\right) \ge m-a_i.
\]
Adding $a_i+1$ to both sides yields
\[
\frac{w_N}{w_i}\cdot (a_i + 1) \ge m+1,
\]
which implies that $a_i \ge q_i + \frac{w_i}{w_N} - 1 > \lfloor q_i\rfloor - 1$.
Since $a_i$ and $\lfloor q_i\rfloor$ are integers, we have $a_i \ge \lfloor q_i\rfloor$, meaning that lower quota is fulfilled.

For the negative part of the theorem, we prove an even stronger claim: we prove that for each $x>0$, there exists an instance in which no WEF$(x,1-x)$ allocation satisfies lower quota;
and for each $x<1$, there exists an instance in which no WEF$(x,1-x)$ allocation satisfies upper quota.
This implies, in addition to the present theorem, that no rule satisfying both quotas can guarantee WEF$(x,1-x)$ for any $x\in[0,1]$. 

Assume first that $x>0$. Choose some integer $n>2/x$, and let $m=n$. The entitlements are $w_1=2$ and $w_2=\cdots=w_n = \frac{n-2}{n-1}$, so that $w_N = n$.
We claim that in every WEF$(x,1-x)$ allocation, every agent receives exactly one item.
Consider an allocation in which some low-entitlement agent gets no item while agent $1$ gets two or more items. 
In order for the WEF$(x,1-x)$ condition for the low-entitlement agent with respect to agent~$1$ to hold, we need
\begin{align*}
\frac{0+(1-x)}{\frac{n-2}{n-1}} \ge \frac{2-x}{2},
\end{align*}
which is equivalent to $2-nx \ge 0$.
Hence, the condition is violated when $n>2/x$.
By a similar reasoning, every allocation in which some agent receives no item while another agent receives two or more items cannot be WEF$(x,1-x)$.
Therefore, in a WEF$(x,1-x)$ allocation, every agent must get exactly one item. But agent~$1$'s quota is $2$, so lower quota is violated.

Assume now that $x<1$. 
Choose some integer $n$ for which
$\frac{n-1}{(n-1)^2-1} < 1-x$ (such an~$n$ exists because the left-hand side approaches $0$ as $n\to\infty$). Let $m=n$, $w_1=n-1$, and $w_2=\cdots=w_n=1/(n-1)$, so that $w_N=n$. 
Consider an allocation in which some low-entitlement agent gets at least one item, while agent~$1$ gets at most $n-1$ items. 
In order for the WEF$(x,1-x)$ condition for agent~$1$ with respect to the low-entitlement agent to hold, we need
\begin{align*}
\frac{(n-1)+(1-x)}{n-1} \ge \frac{1-x}{\frac{1}{n-1}},
\end{align*}
which is equivalent to $\frac{n-1}{(n-1)^2-1} \ge 1-x$.
Hence, the condition is violated when $\frac{n-1}{(n-1)^2-1} < 1-x$.
Therefore, in a WEF$(x,1-x)$ allocation, agent~$1$ must get all $n$ items. 
But agent~$1$'s quota is $n-1$, so  upper quota is violated.
\end{proof}

\begin{theorem}
\label{thm:wprop_quota}
With identical items, for each $x\in[0,1]$, WPROP$(x,1-x)$ implies neither lower quota nor upper quota,
even for two agents with equal entitlements.
\end{theorem}

\begin{proof}
With $m$ identical items, 
the condition for WPROP$(x,y)$ for all $i\in N$ is equivalent to
\begin{align*}
\frac{a_i + y}{w_i} &\ge \frac{m - n x}{w_N},    
\end{align*}
which is the same as
\begin{align}
a_i &\geq \frac{w_i}{w_N}\cdot m - \frac{w_i}{w_N}\cdot n x - y. \label{eq:wpropxy-identical}   
\end{align}
Consider an instance with $n=m=2$ and $w_1=w_2 = 1/2$, and an allocation that gives both items to agent~$2$.
Condition \eqref{eq:wpropxy-identical} is obviously satisfied for agent~$2$, and for agent~$1$ it becomes
\begin{align*}
0\geq 1 - x - y,
\end{align*}
which is satisfied whenever $y=1-x$.
Since both agents have a quota of $1$, both lower and upper quota are violated in this WPROP$(x,1-x)$ allocation.
\end{proof}

\begin{theorem}
\label{thm:mms_quota}
With identical items, for all $n\geq 3$, 
WMMS-fairness and NMMS-fairness imply neither lower quota nor upper quota.
Moreover, there exists an identical-item instance in which a WMMS-fair allocation exists but no allocation providing any positive approximation of WMMS-fairness satisfies lower quota; the same holds for NMMS-fairness.
\end{theorem}

\begin{proof}
For the first statement, let $n\ge 3$ and $m = n-1$. 
Clearly, the WMMS and NMMS of every agent are both $0$ regardless of the weights, so every allocation is WMMS- and NMMS-fair.
However, if all items are given to an agent with the smallest entitlement, upper quota is violated for this agent, since the agent's upper quota is at most $\lceil m/n\rceil = 1$.
Lower quota is also violated in the same allocation as long as there is an agent whose weight is at least $w_N/m$.

We proceed to the second statement.
For WMMS,
let $m=n$ and suppose the entitlements are 
$w_1=2$ and $w_2=\cdots=w_n = \frac{n-2}{n-1}$, so $w_N=n$ and $q_1 = 2$.
Since all agents have the same utility function, a WMMS-fair allocation exists by definition.
Moreover, the WMMS of all agents is positive, since they can partition the items into $n$ bundles with a positive value
(specifically, the WMMS of each agent $i$ is $w_i\cdot \frac{1}{w_1}$, so $\operatorname{WMMS}_1 = 1$ and $\operatorname{WMMS}_2 = \cdots = \operatorname{WMMS}_n = \frac{n-2}{2n-2}$).
Therefore, the unique allocation that provides any positive approximation of WMMS-fairness is the one giving a single item to each agent.
However, in this allocation, the lower quota of agent~$1$ is violated.

For NMMS, let $m=2n-1$ and suppose the entitlements are 
$w_1=n+1$ and $w_2=\cdots=w_n = \frac{n-2}{n-1}$, so  $w_N=2n-1=m$ and $q_1 = w_1=n+1$.
The MMS of every agent is $1$, so 
$\operatorname{NMMS}_1 = \frac{n(n+1)}{2n-1}$, which is less than $n$ for all $n\geq 3$, and $\operatorname{NMMS}_2 = \cdots = \operatorname{NMMS}_n = \frac{n(n-2) }{(2n-1)(n-1)} \in(0,1)$.
There exist NMMS-fair allocations, for example, giving one item to each low-entitlement agent and $n$ items to agent~$1$. 
But every allocation providing a positive approximation to NMMS-fairness must give at least one item to each low-entitlement agent, so the lower quota of agent~$1$ is violated.
\end{proof}

\Cref{thm:mms_quota} requires $n\geq 3$; we show next that this is not a coincidence.

\begin{theorem}
\label{thm:mms_quota_2agents}
With identical items, when $n=2$,
both WMMS-fairness and NMMS-fairness imply both lower and upper quota.
\end{theorem}

\begin{proof}
Let $n = 2$.
Since $a_1 + a_2 = q_1 + q_2 = m$, we have $a_1 \ge \lfloor q_1\rfloor$ if and only if $a_2 \le \lceil q_2 \rceil$.
This means that an allocation either satisfies both lower and upper quotas, or neither.
Hence, it suffices to check one of the two quotas; we will check the lower quota.

Consider the following quantity:
\[
\gamma(A) := \min_{j\in\{1,2\}}\frac{|A_j|}{w_j}.
\]
By definition of WMMS, a WMMS-fair allocation $A$ must maximize $\gamma(A)$.
In more detail, assume that the maximum is $\gamma_{\text{max}}$, and consider an allocation $A'$ with $\gamma(A') < \gamma_{\text{max}}$.
In $A'$, there exists an agent~$j$ with $|A'_j| = \gamma(A')w_j < \gamma_{\text{max}}w_j$.
This agent receives utility less than her WMMS, which is $\gamma_{\text{max}}w_j$, so $A'$ is not WMMS-fair.

Now, consider a WMMS-fair allocation~$A$, which attains $\gamma(A) = \gamma_{\text{max}}$.
If $a_j \le \lfloor q_j\rfloor - 1$ for some $j\in N$, then $a_{3-j} \ge \lceil q_{3-j}\rceil + 1$, so we can obtain a larger $\gamma$ by moving one item from agent~$(3-j)$'s bundle to agent~$j$'s bundle, a contradiction.
Hence, $a_j \ge \lfloor q_j\rfloor$, and lower quota is satisfied.

For NMMS, let $z_i := w_i/w_N$ for each $i\in N$, and note that agent~$i$'s 
MMS is $\floor{m/2}$, so the agent's
NMMS is $2z_i\cdot\floor{m/2} \ge (m-1)z_i$.
Hence, in an NMMS-fair allocation, agent~$i$ receives at least $\lceil (m-1)z_i\rceil$ items.
Since $z_i < 1$, this quantity is at least $\lfloor (m-1)z_i + z_i\rfloor$ = $\lfloor mz_i\rfloor$, which is exactly agent~$i$'s lower quota.
\end{proof}

\begin{theorem}
\label{thm:aps_quota}
With identical items, for all $n\geq 2$,
both OMMS-fairness and APS-fairness are equivalent to
lower quota.
Both notions also imply upper quota when $n=2$, but not when $n\ge 3$.
\end{theorem}

\begin{proof}
We claim that both the OMMS and the APS of each agent $i\in N$ equal $i$'s lower quota $\lfloor q_i\rfloor$.
To see this, in the definition of OMMS, take $\ell = \lfloor q_i\rfloor$ and $d = m$.
With the $d$-partition in which each part contains exactly one item, we find that $\operatorname{MMS}_{i}^{\lfloor q_i\rfloor\text{-out-of-}m}(M) = \lfloor q_i\rfloor$.
Since $\lfloor q_i\rfloor/m \le q_i/m = w_i/w_N$, we have $\operatorname{OMMS}_{i}^{\mathbf{w}}(M) \ge \lfloor q_i\rfloor$.
On the other hand, \citet{BabaioffEzFe21} showed that an agent's OMMS is always at most her APS, which, in turn, is always at most her proportional share; in this case, the proportional share is $q_i$ for agent~$i$.
Since all utilities are integers, the OMMS and APS are both integers, so they both equal $\lfloor q_i\rfloor$.
It follows that OMMS-fairness and APS-fairness are equivalent to
lower quota.

When $n=2$, as in the proof of \Cref{thm:mms_quota_2agents}, upper quota follows immediately from lower quota.
Let $n\ge 3$, $m = n-1$, and suppose that all agents have the same weight.
By the previous paragraph, each agent's OMMS and APS equal the lower quota, which is $0$, so every allocation is OMMS- and APS-fair.
However, if all items are given to the same agent, then upper quota is violated since the agent's upper quota is $1$.
\end{proof}

\begin{theorem}
\label{thm:mwnw_quota}
With identical items, MWNW implies neither lower quota nor upper quota.
\end{theorem}

\begin{proof}
With $m$ identical items, an MWNW allocation corresponds to a partition of $m$ into $a_1+\dots+a_n$ that belongs to
\begin{align*}
\argmax_{
a_1+\dots+a_n=m
}   \prod_{i=1}^{n} a_i^{w_i}
\end{align*}

For lower quota, for each $n\geq 3$, consider an instance with $m=n$ items. 
Every MWNW allocation gives a single item to each agent regardless of the entitlements. 
This violates lower quota for agents with a sufficiently high entitlement---specifically, for every agent~$i$ with $w_i/w_N \ge 2/n$.

For upper quota, consider an instance with $n = 4$ agents, with $w_1=w$, $w_2=w_3=w_4=1$, and $m=4(w+3)/3 = 4 + 4w/3$. 
The value of $w$ will be chosen to ensure that $m$ is an integer.
The upper (and lower) quota of agent~$1$ is $4w/3=m-4$. 
Satisfying agent~$1$'s upper quota means that agents $2$, $3$, and $4$ together must get at least four items. Clearly, in an MWNW allocation, each of them should get at least one item.

\begin{itemize}
\item If together they get (exactly) four items, then the WNW is $(m-4)^w\cdot 2$.
\item If together they get three items, then the WNW is 
$(m-3)^w\cdot 1$.
\end{itemize}

We will choose $w$ so that the latter quantity is larger, that is,
\begin{align*}
(m-3)^w\cdot 1 &> (m-4)^w\cdot 2,
\end{align*}
which is the same as
\begin{align*}
\left(\frac{4w}{3}+1\right)^w &> \left(\frac{4w}{3}\right)^w\cdot 2,    
\end{align*}
or equivalently,
\begin{align*}
\left(1+\frac{3}{4w}\right)^w &> 2.    
\end{align*}
This is satisfied, for example, for $w=6$, in which case $m=12$. 
Agent~$1$'s quota is $8$.
In order to find an MWNW allocation, we consider several cases regarding the value of $a_1$.
\begin{itemize}
\item If $a_1 = 9$, then the maximum possible WNW is $9^6\cdot 1\cdot 1\cdot 1 = 531441$.
\item If $a_1 = 8$, then the maximum possible WNW is $8^6\cdot 1\cdot 1\cdot 2 = 524288$.
\item If $a_1 = 7$, then the maximum possible WNW is $7^6\cdot 1\cdot 2\cdot 2 = 470596$.
\item If $a_1 = 6$, then the maximum possible WNW is $6^6\cdot 2\cdot 2\cdot 2 = 373248$.
\end{itemize}
Other options clearly lead to a smaller weighted Nash welfare. So agent~$1$ gets $9$ items in every MWNW allocation, which is above her upper quota.
\end{proof}

Our examples in \Cref{thm:mwnw_quota} use $3$ and $4$ agents, respectively. 
We show next that for $n=2$ agents, MWNW satisfies both quotas. 
In fact, we prove a slightly more general result.

\begin{theorem}
With identical items, for every pair of agents $i$ and $j$ (possibly among other agents), in every MWNW allocation,
\begin{align*}
\floor{\frac{w_i(a_i+a_j)}{w_i+w_j}} 
\leq 
a_i
\leq 
\ceil{\frac{w_i(a_i+a_j)}{w_i+w_j}}.
\end{align*}
\end{theorem}

\begin{proof}
Fix the bundles allocated to all agents except $i$ and $j$, so that $a_i + a_j$ is fixed too. Denote this sum by $m_0$ and denote $w := \frac{w_i}{w_i+w_j}$. 
From the definition of an MWNW allocation, it follows that for the (fixed) integer $m_0$, $(a_i,a_j)$ is an integer partition of $m_0$ that maximizes the weighted Nash welfare restricted to agents $i$ and $j$. In other words, $a_i$ is an integer maximizing the following function:
\begin{align*}
f(a_i) := a_i^{w}\cdot (m_0-a_i)^{1-w}.
\end{align*}
By taking the derivative of this function, one finds that, if $a_i$ could be any real number, then $f$ is uniquely maximized at $a_i = w\cdot m_0$. Moreover, the function is concave, so it increases up to the maximum and then decreases. Therefore, 
$f(a_i) < f(\floor{w m_0})$ for all $a_i < \floor{w m_0}$, and 
$f(a_i) > f(\ceil{w m_0})$ for all $a_i > \ceil{w m_0}$. 
So the maximum weighted Nash welfare is attained when either $a_i=\floor{w m_0}$ or $a_i=\ceil{w m_0}$, as claimed.
\end{proof}

\begin{corollary}
\label{cor:mwnw_quota_2agents}
With $n=2$ agents and identical items, MWNW satisfies both lower and upper quota.
\end{corollary}

Interestingly, despite failing lower quota in the case of three agents, MWNW satisfies upper quota in this case.\footnote{We are grateful to Ewan Delanoy for help with this proof: https://math.stackexchange.com/a/4271210/29780.}

\begin{theorem}
With $n=3$ agents and identical items, MWNW satisfies upper quota.
\end{theorem}

\begin{proof}
Assume without loss of generality that $w_N = 1$.
If $m\le 2$, each agent gets at most one item in every MWNW allocation due to the tie-breaking convention, and upper quota holds trivially.
Assume therefore that $m\ge 3$.
Suppose that an MWNW allocation assigns $a_i$ items to agent~$i$.
Since the maximum weighted Nash welfare is nonzero, we must have $a_i\ge 1$ for all $i$.
By symmetry, it suffices to show that $a_3 \le \lceil w_3m\rceil$. 
This is clearly true if $a_3 = 1$, so assume that $a_3 \ge 2$.
We will show that $a_3-1 < w_3m$, which implies the desired conclusion.

Since the allocation $(a_1,a_2,a_3)$ yields a higher or equal weighted Nash welfare than $(a_1+1,a_2,a_3-1)$, we have
\[
a_1^{w_1} a_2^{w_2} a_3^{w_3} \ge (a_1+1)^{w_1} a_2^{w_2} (a_3-1)^{w_3}.
\]
This can be rewritten as
\begin{equation}
\label{eq:weight1}
w_1 \le w_3\cdot \frac{\ln\left(\frac{a_3}{a_3-1}\right)}{\ln\left(\frac{a_1+1}{a_1}\right)}.
\end{equation}
Similarly, comparing the allocation $(a_1,a_2,a_3)$ with $(a_1,a_2+1,a_3-1)$, we get
\begin{equation}
\label{eq:weight2}
w_2 \le w_3\cdot \frac{\ln\left(\frac{a_3}{a_3-1}\right)}{\ln\left(\frac{a_2+1}{a_2}\right)}.
\end{equation}
Adding up \eqref{eq:weight1} and \eqref{eq:weight2} and recalling that $w_1+w_2+w_3 = 1$, we get $w_3 \ge S_1$, where
\[
S_1 := \frac{1}{1 + \frac{\ln\left(\frac{a_3}{a_3-1}\right)}{\ln\left(\frac{a_1+1}{a_1}\right)} + \frac{\ln\left(\frac{a_3}{a_3-1}\right)}{\ln\left(\frac{a_2+1}{a_2}\right)}}.
\]
In order to show that $a_3-1 < w_3m$, it suffices to show that $S_2 > 0$, where
\[
S_2 := S_1 - \frac{a_3-1}{m} = \frac{1}{1 + \frac{\ln\left(\frac{a_3}{a_3-1}\right)}{\ln\left(\frac{a_1+1}{a_1}\right)} + \frac{\ln\left(\frac{a_3}{a_3-1}\right)}{\ln\left(\frac{a_2+1}{a_2}\right)}} - \frac{a_3-1}{a_1+a_2+a_3}.
\]

Let us write $z = \frac{a_3}{a_3-1}$, so that $a_3 = \frac{z}{z-1}$; we have $z\in(1,2]$.
We can rewrite $S_2$ as
\[
S_2 = \frac{1}{1 + \frac{\ln z}{\ln\left(\frac{a_1+1}{a_1}\right)} + \frac{\ln z}{\ln\left(\frac{a_2+1}{a_2}\right)}} - \frac{1}{(a_1+a_2+1)z - (a_1+a_2)}.
\]
It therefore suffices to show that $S_3 > 0$, where
\begin{align*}
S_3 
&:= ((a_1+a_2+1)z - (a_1+a_2)) - \left(1 + \frac{\ln z}{\ln\left(\frac{a_1+1}{a_1}\right)} + \frac{\ln z}{\ln\left(\frac{a_2+1}{a_2}\right)}\right) \\
&= (a_1+a_2+1)(z-1) - \left(\frac{\ln z}{\ln\left(\frac{a_1+1}{a_1}\right)} + \frac{\ln z}{\ln\left(\frac{a_2+1}{a_2}\right)}\right) \\
&\ge (a_1+a_2+1)(z-1) - \left(\frac{z-1}{\ln\left(\frac{a_1+1}{a_1}\right)} + \frac{z-1}{\ln\left(\frac{a_2+1}{a_2}\right)}\right) \\
&= (z-1)\left(\left(a_1 - \frac{1}{\ln\left(\frac{a_1+1}{a_1}\right)}\right) + \left(a_2 - \frac{1}{\ln\left(\frac{a_2+1}{a_2}\right)}\right) + 1 \right);
\end{align*}
for the inequality we use the well-known fact that $1+\ln x \le x$, which holds for all $x > 0$.
Hence,
\[
S_3 \ge (z-1)(f(a_1) + f(a_2) + 1),
\]
where $f(x) := x - \frac{1}{\ln\left(\frac{x+1}{x}\right)}$.
Assume without loss of generality that $f(a_1) \le f(a_2)$.
In order to show that $S_3 > 0$, it remains to show that $2f(a_1) + 1 > 0$, that is,
$2m_1 + 1 > \frac{2}{\ln\left(\frac{a_1+1}{a_1}\right)}$.
This is equivalent to $g(a_1) > 0$, where
\[
g(x) := \ln\left(1 + \frac{1}{x}\right) - \frac{2}{2x+1}.
\]
The derivative of $g$ is
\[
g'(x) = -\frac{1}{x(x+1)} + \frac{4}{(2x+1)^2} = -\frac{1}{x(x+1)(2x+1)^2},
\]
which is negative for every $x > 0$.
This means that $g$ is decreasing in the range $[1,\infty)$, and so $g(a_1) > \lim_{x\rightarrow\infty}g(x) = 0$, completing the proof.
\end{proof}

\subsection{Weighted egalitarian rule}
\label{sub:weg}

No notion that we have seen so far guarantees both lower and upper quotas.
In light of this, we introduce a new rule based on the well-known \emph{leximin} principle \citep{Moulin03}. With equal entitlements, leximin aims to maximize the smallest utility, then the second smallest utility, and so on. 
To extend it to the setting with different entitlements, we have to carefully consider how the utilities should be normalized.
A first idea is to normalize the agents' utilities by their proportional share, that is: 
\begin{align*}
\argmax_{(A_1,\dots,A_n)}\leximin_{i\in N} \frac{w_N\cdot u_i(A_i)}{w_i \cdot u_i(M)},
\end{align*}
where ``leximin'' means that we maximize the smallest value, then the second smallest, and so on.
This idea is similar to the one used in the asymmetric leximin bargaining solution of \citet{driesen2012asymmetric}.
However, the resulting rule yields a blatantly unfair outcome when there is a single item and two agents: since the smallest utility is always $0$, the item is allocated to the agent with a \emph{smaller} entitlement, as this makes the ratio $w_N/w_i$ larger. 
Therefore, we instead maximize the minimum \emph{difference} between the agents' normalized utilities and their proportional shares.

\begin{definition}[WEG]
A \emph{weighted egalitarian} allocation is an allocation in 
\begin{align*}
\argmax_{(A_1,\dots,A_n)} \leximin_{i\in N} \bigg(\frac{u_i(A_i)}{u_i(M)} - \frac{w_i}{w_N}\bigg).
\end{align*}
\end{definition}

This rule is similar to the  ``Leximin rule'' for apportionment \citep{biro2015fair}; the latter rule minimizes the leximin vector of the \emph{departures}, defined as the absolute difference $\left|\frac{w_i\cdot  u_i(M)}{w_N \cdot u_i(A_i)} - 1\right|$.

Note that in the instance with one item and two agents, a WEG allocation gives the item to the agent with a larger entitlement, as one would reasonably expect. Moreover, in \Cref{ex:quotas}, 
the WEG allocations are $(2,1,0)$ and $(2,0,1)$, which satisfy both quotas. The next theorem shows that this holds in general.

\begin{theorem}
\label{thm:weg_quota}
With identical items, every WEG allocation satisfies both lower and upper quota.
\end{theorem}

\begin{proof}
With $m$ identical items, a WEG allocation corresponds to a partition of $m$ into $a_1+\dots+a_n$ that belongs to
\begin{align*}
\argmax_{
a_1+\dots+a_n=m
}   \leximin_{i\in N}  \left(a_i - q_i\right),
\end{align*}
where $q_i = \frac{w_i}{w_N}\cdot m$.
Let $d_i := a_i - q_i$.
For every agent $i$:
\begin{itemize}
\item $d_i > -1$ if and only if $i$ gets at least her lower quota $\floor{q_i}$;
\item $d_i < 1$ if and only if $i$ gets at most her upper quota 
$\ceil{q_i}$.
\end{itemize}
There exists an allocation in which every agent gets at least her lower quota, so $\min_{i\in N} d_i > -1$ in such an allocation. 
Therefore, the same must hold for a WEG allocation, so this allocation satisfies lower quota for all agents.

If upper quota is violated for some agent $i$, then $d_i \ge 1$ and there is another agent $j$ for whom $d_j<0$. 
If a single item is moved from $i$ to $j$, then  $d_i$ remains nonnegative while $d_j$ increases, so the leximin vector strictly improves. Therefore, a WEG allocation satisfies upper quota for all agents.
\end{proof}

\begin{remark}
\label{rem:weg}
By the negative part of the proof of Theorem~\ref{thm:wef_quota}, 
since WEG satisfies both quotas, it cannot imply WEF$(x,1-x)$ for any $x\in[0,1]$.
Similarly, 
by Theorem~\ref{thm:mms_quota}, it cannot imply any positive approximation to WMMS or NMMS.

On the positive side, since WEG satisfies lower quota,
by Theorem \ref{thm:aps_quota}, WEG implies APS-fairness and OMMS-fairness.
In Appendix~\ref{app:WEG-binary},
we show that this implication holds even in the more general setting of \emph{binary} valuations.
\end{remark}

In contrast to share-based notions, a WEG allocation always exists by definition.
Moreover, it is relatively easy to explain such an allocation. 
For example, 
if a WEG allocation  in a particular instance yields a minimum difference  $\frac{u_i(A_i)}{u_i(M)} - \frac{w_i}{w_N} = -0.05$, one can explain to the agents that ``each of you receives only $5\%$ less than your proportional share (or better), and there is no allocation with a smaller deviation''. 
This makes the egalitarian approach attractive for further study in the context of unequal entitlements.

\section{Experiments}
\label{sec:experiments}
As we saw in Section~\ref{sec:fairness-up-to-1}, for each $y\in[0,1]$, the divisor picking sequence with function $f(t) = t+y$ satisfies WEF$(1-y,y)$ and WPROP$(1-y,y)$, and is the only divisor picking sequence to do so.
Picking sequences with smaller $y$ favor lower-entitlement agents, while those with larger $y$ treat higher-entitlement agents better.
The choice of picking sequence could therefore significantly influence the utilities that different agents receive from the resulting allocation, and an argument can be made for each choice using the corresponding WEF$(1-y,y)$ and WPROP$(1-y,y)$ notions, which are approximations to weighted envy-freeness and weighted proportionality, respectively.
In this section, we empirically study the likelihood that a picking sequence produces an allocation satisfying one of the following benchmarks \emph{exactly}: weighted envy-freeness, weighted proportionality, WMMS-fairness, and NMMS-fairness.
This provides us with a family of criteria for judging which picking sequence may be considered ``fairer'' than others.

For our experiments, we use the parameters $n = 3$ and $m\in\{6,8,10\}$.
We examine two distributions from which the utility of each agent for each item is independently drawn: the uniform distribution on the interval $[0,1]$ and the exponential distribution with mean~$1$.
The agents' weights are drawn uniformly at random from $[0,1]$.
We consider picking sequences $f(t) = t + y$ for $y\in\{0, 0.05, 0.1, \dots, 0.95, 1\}$; in case of ties in the highest ratio $f(t_i)/w_i$, we break ties in favor of larger $w_i$.
For each of these combinations and each fairness benchmark, we generated a number of instances and recorded the percentage of returned allocations that meet the benchmark.
In particular, for weighted envy-freeness and weighted proportionality we generated 100,000 instances, while for WMMS-fairness and NMMS-fairness we produced 20,000 instances due to the high running time required to compute the agents' WMMS and NMMS.
The results are shown in Figures~\ref{fig:WEF}--\ref{fig:NMMS}, and our code can be found at \url{https://github.com/erelsgl/fairness-with-entitlements}.

Intuitively, since medium values of $y$ such as $y = 0.5$ provide a natural middle ground between satisfying high-weight and low-weight agents, one may think that such values should perform optimally with respect to common fairness benchmarks.
As Figure~\ref{fig:WEF} shows, this is largely true when the benchmark is weighted envy-freeness: the curves first rise and then fall as one moves from $y = 0$ to $y = 1$.
Note also that the percentage of returned allocations that are weighted envy-free clearly increases as $m$ grows---with a larger number of items, there is more opportunity for any envy that arises to be eliminated as the agents pick their favorite items.
Furthermore, the percentage of weighted envy-free allocations is higher when the utilities are drawn from the exponential distribution than when they are drawn from the uniform distribution.
A possible explanation is that the exponential distribution allows a greater chance that, for each agent, some item is highly valuable to the agent (from the ``long tail'' of the distribution) and not so valuable to the remaining agents; allocating such an item to the agent immediately eliminates envy.

On the other hand---and perhaps surprisingly---as can be seen in Figures~\ref{fig:WPROP}--\ref{fig:NMMS}, picking sequences with low values of~$y$ fare better when it comes to weighted proportionality, WMMS-fairness, and NMMS-fairness. 
Given that these picking sequences favor lower-weight agents, our results indicate that satisfying such agents is more helpful for
fairness benchmarks based on a utility threshold for each agent 
than it is for weighted envy-freeness.
Providing a theoretical explanation of this phenomenon is an interesting direction for future research.\footnote{Theorems~\ref{thm:wef10_nmms} and \ref{thm:wefx1mx_nmms} provide a partial explanation concerning NMMS: the only $y$ for which WEF$(1-y,y)$, which is satisfied by the divisor picking sequence with $f(t) = t+y$, implies a positive approximation of NMMS is $y=0$.}

\begin{figure}
	\centering
	\begin{subfigure}{.5\textwidth} 
		\centering
\begin{tikzpicture}[scale=0.85]
\begin{axis}[
    title={(a) Uniform distribution on $[0,1]$},
    xlabel={},
    ylabel={},
    xmin=0, xmax=1,
    ymin=0, ymax=100,
    xtick= {0,0.2,0.4,0.6,0.8,1},
    ytick={0,20,40,60,80,100},
    legend pos=north west,
    ymajorgrids=true,
		xmajorgrids=true,
    grid style=dashed,
]

\addplot[
	    color=black,
	    mark=triangle*,
	    mark size=3,
		fill opacity=0.8,
  	    draw opacity=0.8,
	    ]
	    coordinates {
	    (0.0,26.773)(0.05,27.199)(0.1,27.52)(0.15,27.953)(0.2,28.278)(0.25,28.734)(0.3,28.999)(0.35,29.369)(0.4,29.699)(0.45,29.766)(0.5,29.859)(0.55,29.872)(0.6,29.404)(0.65,28.392)(0.7,27.783)(0.75,27.059)(0.8,26.305)(0.85,25.445)(0.9,24.215)(0.95,23.15)(1.0,22.331)
	    };
	    \addlegendentry{$m=6$}

\addplot[
    color=red,
    mark=square*,
		fill opacity=0.8,
  	draw opacity=0.8,
    ]
    coordinates {
   (0.0,40.826)(0.05,41.457)(0.1,42.181)(0.15,43.161)(0.2,44.23)(0.25,44.979)(0.3,46.049)(0.35,47.418)(0.4,48.019)(0.45,48.657)(0.5,48.658)(0.55,48.361)(0.6,48.28)(0.65,47.385)(0.7,46.348)(0.75,45.155)(0.8,43.73)(0.85,42.243)(0.9,40.839)(0.95,39.495)(1.0,37.372)
    };
    \addlegendentry{$m=8$}
    
\addplot[
    color=blue,
    mark=otimes*,
		fill opacity=0.8,
  	draw opacity=0.8,
    ]
    coordinates {
    (0.0,52.167)(0.05,52.926)(0.1,54.358)(0.15,55.352)(0.2,56.368)(0.25,58.164)(0.3,59.333)(0.35,60.762)(0.4,62.043)(0.45,62.771)(0.5,63.169)(0.55,63.404)(0.6,63.105)(0.65,62.462)(0.7,61.561)(0.75,60.162)(0.8,58.362)(0.85,56.609)(0.9,54.823)(0.95,52.596)(1.0,50.743)
    };
    \addlegendentry{$m=10$}
\end{axis}
\end{tikzpicture}
\end{subfigure}\begin{subfigure}{.5\textwidth} 
	\centering
\begin{tikzpicture}[scale=0.85]
\begin{axis}[
    title={(b) Exponential distribution with mean $1$},
    xlabel={},
    ylabel={},
    xmin=0, xmax=1,
    ymin=0, ymax=100,
    xtick= {0,0.2,0.4,0.6,0.8,1},
    ytick={0,20,40,60,80,100},
    legend pos=south west,
    ymajorgrids=true,
		xmajorgrids=true,
    grid style=dashed,
]

\addplot[
	    color=black,
	    mark size=3,
        mark=triangle*,
        fill opacity=0.8,
        draw opacity=0.8,
	    ]
	    coordinates {
	    (0.0,38.942)(0.05,38.983)(0.1,39.428)(0.15,39.91)(0.2,40.248)(0.25,40.184)(0.3,40.701)(0.35,40.561)(0.4,39.76)(0.45,39.293)(0.5,38.287)(0.55,37.061)(0.6,36.532)(0.65,35.02)(0.7,33.597)(0.75,32.345)(0.8,31.069)(0.85,29.932)(0.9,28.788)(0.95,27.826)(1.0,26.317)
    	   };
	    \addlegendentry{$m=6$}

\addplot[
    color=red,
    mark=square*,
    fill opacity=0.8,
  	draw opacity=0.8,
    ]
    coordinates {
    (0.0,53.234)(0.05,53.613)(0.1,54.286)(0.15,55.142)(0.2,55.945)(0.25,56.8)(0.3,57.578)(0.35,57.816)(0.4,57.557)(0.45,56.906)(0.5,56.114)(0.55,54.923)(0.6,53.858)(0.65,52.815)(0.7,50.643)(0.75,49.135)(0.8,47.544)(0.85,46.131)(0.9,44.581)(0.95,42.933)(1.0,41.536)
    };
    \addlegendentry{$m=8$}
    
\addplot[
    color=blue,
    mark=otimes*,
	fill opacity=0.8,
  	draw opacity=0.8,
    ]
    coordinates {
    (0.0,63.459)(0.05,64.142)(0.1,64.628)(0.15,65.793)(0.2,67.198)(0.25,68.428)(0.3,69.186)(0.35,69.712)(0.4,69.769)(0.45,69.537)(0.5,68.882)(0.55,67.623)(0.6,66.64)(0.65,65.32)(0.7,63.611)(0.75,62.136)(0.8,60.092)(0.85,58.491)(0.9,57.085)(0.95,55.26)(1.0,53.583)
    };
    \addlegendentry{$m=10$}
\end{axis}
\end{tikzpicture}
\end{subfigure}

\caption{Percentage of allocations produced by the divisor picking sequence with function $f(t) = t+y$ that satisfy weighted envy-freeness for three agents. The $x$-axis indicates the value of $y$ in the function.}
\label{fig:WEF}
\end{figure}

\begin{figure}
	\centering
	\begin{subfigure}{.5\textwidth} 
		\centering
\begin{tikzpicture}[scale=0.85]
\begin{axis}[
    title={(a) Uniform distribution on $[0,1]$},
    xlabel={},
    ylabel={},
    xmin=0, xmax=1,
    ymin=0, ymax=100,
    xtick= {0,0.2,0.4,0.6,0.8,1},
    ytick={0,20,40,60,80,100},
    legend pos=south west,
    ymajorgrids=true,
		xmajorgrids=true,
    grid style=dashed,
]

\addplot[
	    color=black,
	    mark size=3,
	    mark=triangle*,
			fill opacity=0.8,
  	  draw opacity=0.8,
	    ]
	    coordinates {
	    (0.0,66.414)(0.05,66.22)(0.1,66.514)(0.15,65.957)(0.2,65.482)(0.25,64.416)(0.3,63.693)(0.35,62.043)(0.4,60.957)(0.45,59.311)(0.5,57.318)(0.55,55.455)(0.6,53.113)(0.65,51.145)(0.7,48.99)(0.75,46.971)(0.8,44.806)(0.85,42.713)(0.9,41.256)(0.95,39.278)(1.0,37.847)
	    };
	    \addlegendentry{$m=6$}

\addplot[
    color=red,
    mark=square*,
		fill opacity=0.8,
  	draw opacity=0.8,
    ]
    coordinates {
    (0.0,82.117)(0.05,82.1)(0.1,82.419)(0.15,81.844)(0.2,81.463)(0.25,80.394)(0.3,79.801)(0.35,78.268)(0.4,77.264)(0.45,75.542)(0.5,73.517)(0.55,71.677)(0.6,69.419)(0.65,67.483)(0.7,65.137)(0.75,62.719)(0.8,60.629)(0.85,58.516)(0.9,56.64)(0.95,54.239)(1.0,52.256)
    };
    \addlegendentry{$m=8$}
    
\addplot[
    color=blue,
    mark=otimes*,
		fill opacity=0.8,
  	draw opacity=0.8,
    ]
    coordinates {
    (0.0,90.04)(0.05,89.984)(0.1,89.814)(0.15,89.336)(0.2,88.964)(0.25,88.054)(0.3,86.978)(0.35,86.062)(0.4,84.977)(0.45,83.546)(0.5,81.872)(0.55,80.006)(0.6,78.218)(0.65,76.532)(0.7,74.419)(0.75,72.693)(0.8,70.466)(0.85,68.537)(0.9,66.579)(0.95,64.491)(1.0,62.654)
    };
    \addlegendentry{$m=10$}
\end{axis}
\end{tikzpicture}
\end{subfigure}\begin{subfigure}{.5\textwidth} 
	\centering
\begin{tikzpicture}[scale=0.85]
\begin{axis}[
    title={(b) Exponential distribution with mean $1$},
    xlabel={},
    ylabel={},
    xmin=0, xmax=1,
    ymin=0, ymax=100,
    xtick= {0,0.2,0.4,0.6,0.8,1},
    ytick={0,20,40,60,80,100},
    legend pos=south west,
    ymajorgrids=true,
		xmajorgrids=true,
    grid style=dashed,
]

\addplot[
	    color=black,
	    mark size=3,
        mark=triangle*,
        fill opacity=0.8,
        draw opacity=0.8,
	    ]
	    coordinates {
	    (0.0,75.283)(0.05,74.143)(0.1,72.906)(0.15,71.399)(0.2,69.924)(0.25,67.536)(0.3,65.283)(0.35,63.473)(0.4,61.091)(0.45,58.687)(0.5,56.204)(0.55,53.804)(0.6,51.865)(0.65,49.685)(0.7,47.704)(0.75,45.466)(0.8,43.513)(0.85,41.979)(0.9,40.137)(0.95,38.71)(1.0,37.269)
    	   };
	    \addlegendentry{$m=6$}

\addplot[
    color=red,
    mark=square*,
    fill opacity=0.8,
  	draw opacity=0.8,
    ]
    coordinates {
    (0.0,87.595)(0.05,86.749)(0.1,85.671)(0.15,84.452)(0.2,83.321)(0.25,81.459)(0.3,79.733)(0.35,77.724)(0.4,75.997)(0.45,73.62)(0.5,71.496)(0.55,69.149)(0.6,66.76)(0.65,64.767)(0.7,62.537)(0.75,60.68)(0.8,58.86)(0.85,57.209)(0.9,55.094)(0.95,53.486)(1.0,51.777)
    };
    \addlegendentry{$m=8$}
    
\addplot[
    color=blue,
    mark=otimes*,
		fill opacity=0.8,
  	draw opacity=0.8,
    ]
    coordinates {
    (0.0,93.127)(0.05,92.465)(0.1,91.534)(0.15,90.64)(0.2,89.351)(0.25,88.246)(0.3,86.909)(0.35,85.144)(0.4,83.43)(0.45,81.788)(0.5,80.02)(0.55,77.943)(0.6,76.227)(0.65,74.341)(0.7,72.439)(0.75,70.881)(0.8,69.255)(0.85,66.984)(0.9,65.375)(0.95,63.856)(1.0,62.465)
    };
    \addlegendentry{$m=10$}
\end{axis}
\end{tikzpicture}
\end{subfigure}

\caption{Percentage of allocations produced by the divisor picking sequence with function $f(t) = t+y$ that satisfy weighted proportionality for three agents. The $x$-axis indicates the value of $y$ in the function.}
\label{fig:WPROP}
\end{figure}

\begin{figure}
	\centering
	\begin{subfigure}{.5\textwidth} 
		\centering
\begin{tikzpicture}[scale=0.85]
\begin{axis}[
    title={(a) Uniform distribution on $[0,1]$},
    xlabel={},
    ylabel={},
    xmin=0, xmax=1,
    ymin=0, ymax=100,
    xtick= {0,0.2,0.4,0.6,0.8,1},
    ytick={0,20,40,60,80,100},
    legend pos=south west,
    ymajorgrids=true,
		xmajorgrids=true,
    grid style=dashed,
]

\addplot[
	    color=black,
	    mark size=3,
	    mark=triangle*,
			fill opacity=0.8,
  	  draw opacity=0.8,
	    ]
	    coordinates {
	    (0.0,78.79)(0.05,78.41)(0.1,77.625)(0.15,76.25)(0.2,75.31)(0.25,73.595)(0.3,72.275)(0.35,69.655)(0.4,67.715)(0.45,65.55)(0.5,62.625)(0.55,60.405)(0.6,58.48)(0.65,55.595)(0.7,53.91)(0.75,51.315)(0.8,49.4)(0.85,46.35)(0.9,45.305)(0.95,43.235)(1.0,41.695)
	    };
	    \addlegendentry{$m=6$}

\addplot[
    color=red,
    mark=square*,
		fill opacity=0.8,
  	draw opacity=0.8,
    ]
    coordinates {
    (0.0,86.55)(0.05,86.315)(0.1,85.575)(0.15,84.935)(0.2,84.03)(0.25,83.365)(0.3,81.825)(0.35,80.32)(0.4,78.835)(0.45,76.775)(0.5,74.81)(0.55,72.62)(0.6,70.105)(0.65,68.23)(0.7,66.11)(0.75,63.835)(0.8,61.5)(0.85,59.74)(0.9,57.165)(0.95,55.345)(1.0,53.095)
    };
    \addlegendentry{$m=8$}
    
\addplot[
    color=blue,
    mark=otimes*,
		fill opacity=0.8,
  	draw opacity=0.8,
    ]
    coordinates {
    (0.0,91.52)(0.05,91.39)(0.1,91.19)(0.15,90.68)(0.2,90.015)(0.25,88.72)(0.3,88.395)(0.35,86.85)(0.4,85.53)(0.45,84.115)(0.5,82.035)(0.55,80.065)(0.6,78.99)(0.65,76.935)(0.7,74.345)(0.75,72.08)(0.8,70.43)(0.85,68.435)(0.9,66.635)(0.95,64.675)(1.0,62.24)
    };
    \addlegendentry{$m=10$}
\end{axis}
\end{tikzpicture}
\end{subfigure}\begin{subfigure}{.5\textwidth} 
	\centering
\begin{tikzpicture}[scale=0.85]
\begin{axis}[
    title={(b) Exponential distribution with mean $1$},
    xlabel={},
    ylabel={},
    xmin=0, xmax=1,
    ymin=0, ymax=100,
    xtick= {0,0.2,0.4,0.6,0.8,1},
    ytick={0,20,40,60,80,100},
    legend pos=south west,
    ymajorgrids=true,
		xmajorgrids=true,
    grid style=dashed,
]

\addplot[
	    color=black,
	    mark size=3,
        mark=triangle*,
        fill opacity=0.8,
        draw opacity=0.8,
	    ]
	    coordinates {
	    (0.0,83.915)(0.05,82.34)(0.1,80.595)(0.15,78.52)(0.2,76.185)(0.25,73.975)(0.3,71.69)(0.35,68.635)(0.4,65.685)(0.45,63.11)(0.5,61.105)(0.55,58.55)(0.6,56.465)(0.65,54.03)(0.7,51.27)(0.75,49.88)(0.8,47.27)(0.85,45.515)(0.9,44.215)(0.95,42.465)(1.0,40.91)
    	   };
	    \addlegendentry{$m=6$}

\addplot[
    color=red,
    mark=square*,
    fill opacity=0.8,
  	draw opacity=0.8,
    ]
    coordinates {
    (0.0,90.0)(0.05,88.8)(0.1,88.325)(0.15,86.68)(0.2,85.3)(0.25,83.78)(0.3,81.445)(0.35,79.27)(0.4,77.49)(0.45,74.695)(0.5,72.44)(0.55,70.585)(0.6,67.785)(0.65,65.925)(0.7,63.845)(0.75,62.19)(0.8,59.745)(0.85,57.54)(0.9,56.32)(0.95,54.59)(1.0,52.745)
    };
    \addlegendentry{$m=8$}
    
\addplot[
    color=blue,
    mark=otimes*,
		fill opacity=0.8,
  	draw opacity=0.8,
    ]
    coordinates {
    (0.0,94.125)(0.05,93.265)(0.1,92.48)(0.15,91.185)(0.2,89.915)(0.25,88.57)(0.3,87.17)(0.35,85.81)(0.4,83.725)(0.45,82.235)(0.5,79.985)(0.55,78.34)(0.6,76.66)(0.65,74.86)(0.7,72.945)(0.75,70.61)(0.8,68.865)(0.85,68.005)(0.9,65.915)(0.95,64.005)(1.0,62.08)
    };
    \addlegendentry{$m=10$}
\end{axis}
\end{tikzpicture}
\end{subfigure}

\caption{Percentage of allocations produced by the divisor picking sequence with function $f(t) = t+y$ that satisfy WMMS-fairness for three agents. The $x$-axis indicates the value of~$y$ in the function.}
\label{fig:WMMS}
\end{figure}

\begin{figure}
	\centering
	\begin{subfigure}{.5\textwidth} 
		\centering
\begin{tikzpicture}[scale=0.85]
\begin{axis}[
    title={(a) Uniform distribution on $[0,1]$},
    xlabel={},
    ylabel={},
    xmin=0, xmax=1,
    ymin=0, ymax=100,
    xtick= {0,0.2,0.4,0.6,0.8,1},
    ytick={0,20,40,60,80,100},
    legend pos=south west,
    ymajorgrids=true,
		xmajorgrids=true,
    grid style=dashed,
]

\addplot[
	    color=black,
	    mark size=3,
	    mark=triangle*,
			fill opacity=0.8,
  	  draw opacity=0.8,
	    ]
	    coordinates {
	    (0.0,83.185)(0.05,81.815)(0.1,80.7)(0.15,79.595)(0.2,78.22)(0.25,76.56)(0.3,74.685)(0.35,71.77)(0.4,70.24)(0.45,67.43)(0.5,65.515)(0.55,63.24)(0.6,61.565)(0.65,59.325)(0.7,56.94)(0.75,54.21)(0.8,51.755)(0.85,50.165)(0.9,48.655)(0.95,46.39)(1.0,45.08)
	    };
	    \addlegendentry{$m=6$}

\addplot[
    color=red,
    mark=square*,
		fill opacity=0.8,
  	draw opacity=0.8,
    ]
    coordinates {
    (0.0,86.875)(0.05,87.015)(0.1,86.47)(0.15,86.19)(0.2,85.17)(0.25,84.21)(0.3,82.87)(0.35,81.585)(0.4,79.185)(0.45,77.14)(0.5,75.04)(0.55,73.745)(0.6,71.64)(0.65,69.22)(0.7,67.375)(0.75,65.28)(0.8,62.125)(0.85,60.555)(0.9,58.19)(0.95,56.52)(1.0,54.67)
    };
    \addlegendentry{$m=8$}
    
\addplot[
    color=blue,
    mark=otimes*,
		fill opacity=0.8,
  	draw opacity=0.8,
    ]
    coordinates {
    (0.0,91.34)(0.05,91.055)(0.1,90.755)(0.15,90.235)(0.2,89.505)(0.25,88.91)(0.3,88.155)(0.35,86.625)(0.4,85.2)(0.45,83.985)(0.5,82.405)(0.55,80.665)(0.6,79.035)(0.65,76.635)(0.7,74.96)(0.75,72.96)(0.8,71.325)(0.85,68.74)(0.9,67.365)(0.95,65.88)(1.0,63.42)
    };
    \addlegendentry{$m=10$}
\end{axis}
\end{tikzpicture}
\end{subfigure}\begin{subfigure}{.5\textwidth} 
	\centering
\begin{tikzpicture}[scale=0.85]
\begin{axis}[
    title={(b) Exponential distribution with mean $1$},
    xlabel={},
    ylabel={},
    xmin=0, xmax=1,
    ymin=0, ymax=100,
    xtick= {0,0.2,0.4,0.6,0.8,1},
    ytick={0,20,40,60,80,100},
    legend pos=south west,
    ymajorgrids=true,
		xmajorgrids=true,
    grid style=dashed,
]

\addplot[
	    color=black,
	    mark size=3,
        mark=triangle*,
        fill opacity=0.8,
        draw opacity=0.8,
	    ]
	    coordinates {
	    (0.0,89.335)(0.05,87.52)(0.1,85.495)(0.15,83.53)(0.2,80.585)(0.25,78.645)(0.3,75.31)(0.35,73.45)(0.4,70.6)(0.45,67.9)(0.5,65.54)(0.55,62.595)(0.6,60.82)(0.65,58.07)(0.7,55.89)(0.75,53.315)(0.8,51.64)(0.85,49.455)(0.9,48.315)(0.95,46.095)(1.0,43.915)
    	   };
	    \addlegendentry{$m=6$}

\addplot[
    color=red,
    mark=square*,
    fill opacity=0.8,
  	draw opacity=0.8,
    ]
    coordinates {
    (0.0,92.66)(0.05,91.325)(0.1,90.205)(0.15,88.605)(0.2,86.93)(0.25,85.575)(0.3,83.245)(0.35,81.18)(0.4,79.3)(0.45,76.805)(0.5,74.435)(0.55,72.215)(0.6,70.155)(0.65,67.85)(0.7,65.63)(0.75,64.5)(0.8,62.03)(0.85,60.06)(0.9,58.415)(0.95,57.475)(1.0,55.315)
    };
    \addlegendentry{$m=8$}
    
\addplot[
    color=blue,
    mark=otimes*,
		fill opacity=0.8,
  	draw opacity=0.8,
    ]
    coordinates {
    (0.0,94.745)(0.05,94.04)(0.1,93.19)(0.15,91.775)(0.2,90.87)(0.25,89.625)(0.3,87.725)(0.35,86.44)(0.4,84.64)(0.45,83.07)(0.5,81.665)(0.55,79.26)(0.6,77.67)(0.65,75.27)(0.7,74.295)(0.75,72.295)(0.8,70.6)(0.85,68.4)(0.9,66.715)(0.95,64.955)(1.0,64.03)
    };
    \addlegendentry{$m=10$}
\end{axis}
\end{tikzpicture}
\end{subfigure}

\caption{Percentage of allocations produced by the divisor picking sequence with function $f(t) = t+y$ that satisfy NMMS-fairness for three agents. The $x$-axis indicates the value of~$y$ in the function.}
\label{fig:NMMS}
\end{figure}

\section{Conclusion and Future Work}

In this paper, we revisited known fairness notions for the setting where agents can have different entitlements to the resource, and introduced several new notions for this setting.
Our work further reveals the richness of weighted fair division that has been uncovered by a number of recent papers.
Indeed, when all agents have the same weight, WEF$(x,1-x)$ reduces to EF1 and WPROP$(x,1-x)$ reduces to PROP1 for all $x\in[0,1]$, while WMMS, NMMS, and OMMS all reduce to MMS.
We believe that the concepts we have introduced add meaningful value beyond those proposed in prior work.
In particular, the notions WEF$(x,1-x)$ and WPROP$(x,1-x)$ allow us to choose the degree to which we want to prioritize agents with larger weights in comparison to those with smaller weights.
Furthermore, our NMMS notion provides an intuitive generalization of the well-studied MMS criterion for which a nontrivial approximation can always be attained.

Given the plethora of fairness notions in the weighted setting---several of which we showed are incompatible with, or do not imply any approximation of, one another---it should come as no surprise that the choice of notion plays a key role when deciding which rules or allocations are fairer than others.
Indeed, this is clearly illustrated in our experimental results in Section~\ref{sec:experiments}, which indicate that the fairest picking sequences according to weighted envy-freeness are different from the ones according to weighted proportionality and maximin share-based notions.
In light of this, one should be careful about ``cherry picking'', that is, justifying the use of certain algorithms or outcomes based only on specific (and perhaps favorably chosen) fairness notions.
An avenue for future research is to conduct experiments with human subjects---this would demonstrate which fairness notions are considered fairer by humans, similarly to the ``moral machine'' experiment \citep{awad2018moral} and various earlier experiments in fair division
\citep{dickinson2002fair,engelmann2004inequality,fehr2006inequality,herreiner2007distributing}.

On the theoretical front, a possible future work direction is to study our new fairness notions in conjunction with other properties, for example, 
the economic efficiency property of 
\emph{Pareto optimality}.
\citet{ChakrabortyIgSu20} have shown, by means of generalizing \citet{BarmanKrVa18}'s market-based algorithm, that WEF$(1,0)$ is compatible with Pareto optimality, but it remains unclear whether their argument can be further generalized to work for WEF$(x,1-x)$ when $x < 1$.
Considering the issue of strategic manipulation could also lead to useful insights---even in the unweighted setting, previous work has revealed that picking sequences do not perform well in view of strategyproofness \citep{BouveretLa14,TominagaToYo16,XiaoLi20}.
Further, one could try to turn our non-implication findings into incompatibility results,\footnote{
An example is the discussion after \Cref{thm:wefx1mx_nmms}, which  shows that WEF$(x,1-x)$ for any $x\in[0,1)$ is incompatible with any positive approximation of NMMS.
}
or to demonstrate that the notions that do not imply each other are in fact compatible.
Finally, extending some of our results to non-additive utilities is a challenging but important direction as well; a step in that direction was taken in the follow-up work of \citet{MontanariScSu24}, which focused on submodular utilities.

\subsection*{Acknowledgments}

This work was partially supported by the Israel Science Foundation under grant number 712/20, by the Singapore Ministry of Education under grant number MOE-T2EP20221-0001, and by an NUS Start-up Grant.
A preliminary version of the work appeared in Proceedings
of the 36th AAAI Conference on Artificial Intelligence, February--March 2022.
We would like to thank the anonymous reviewers of AAAI 2022 and ACM Transactions on Economics and Computation for their valuable comments.

\bibliographystyle{ACM-Reference-Format}
\bibliography{main}

\appendix
\section{Omitted Proofs}
\label{app:omitted}

\subsection*{Proof of \Cref{thm:picking-WEF}}

Our proof follows the template laid out in Theorem~3.3 of \citet{ChakrabortyIgSu20}.
For notational convenience, let $y := 1-x$.

($\Rightarrow$) 
Assume that $\pi$ fulfills \wefxy.
In particular, 
for every integer $k\geq 1$,
\wefxy is satisfied in an instance where every agent has utility $1$ for some $k$ items and $0$ for the remaining $m-k$ items.
The $k$ valuable items are picked in the first $k$ turns, so if we define $t_i$ and $t_j$ based on this prefix, agent $i$ has utility $t_i$ for bundle $A_i$ and utility $t_j$ for bundle~$A_j$.
If $t_j=0$ then the claim trivially holds. Otherwise, the most valuable item $g$ in $A_j$ has a value of $1$, so by definition of \wefxy,
\begin{align*}
\frac{t_i+y}{w_i}
=
\frac{u_i(A_i) + y\cdot u_i(g)}{w_i} 
\geq
\frac{u_i(A_j) - x\cdot u_i(g)}{w_j}
=
\frac{t_j-x}{w_j}.
\end{align*}
It follows that $t_i + y \geq \frac{w_i}{w_j}(t_j-x)$.

($\Leftarrow$)
Consider any pair of agents $i,j$.
It suffices to show that the WEF$(x,y)$ condition for agent $i$ towards agent $j$ is fulfilled after every pick by agent $j$.
Consider any pick by agent~$j$---suppose it is the agent's $t_j$-th pick.
We divide the sequence of picks up to this pick into \emph{phases}, where each phase $\ell\in[t_j]$ consists of the picks after agent~$j$'s $(\ell-1)$st pick up to (and including) the agent's $\ell$-th pick. 
We denote:
\begin{itemize}
\item $\tau_\ell := $ the number of times agent $i$ picks in phase $\ell$ (that is, between agent $j$'s
 $(\ell-1)$st and $\ell$-th picks);
\item 
$\alpha_1^\ell, \alpha_2^\ell, \dots,\alpha_{\tau_\ell}^\ell := $
agent $i$'s utilities for the items that she picks herself in phase $\ell$ (if any);
\item 
$\alpha_\ell :=  \alpha_1^\ell + \cdots + \alpha_{\tau_\ell}^\ell = $ the total utility gained by agent $i$ in phase $\ell$;
\item 
$\beta_\ell := $ agent $i$'s utility for the item that agent $j$ picks at the end of phase $\ell$.
\end{itemize}

Let $\rho := w_i/w_j$. 
For each integer $s\in[t_j]$,
applying the condition in the theorem statement to the picking sequence up to and including phase $s$, we have
\begin{equation}
\label{eq:WWEF-picks}
y+\sum_{\ell=1}^s \tau_\ell ~\ge~ \rho(s-x)
\qquad\qquad \forall s\in [t_j].
\end{equation}
Every time agent $i$ picks, she picks an item with the highest utility for her among the available items. 
In particular, in each phase $\ell$, she picks $\tau_{\ell}$ items each of which yields at least as high utility to her as each item not yet picked by agent $j$. This implies
\begin{equation}
\label{eq:WWEF-utils}
\alpha_\ell
\geq \tau_\ell \cdot \max_{\ell\le r \le t_j} \beta_r
\qquad\qquad \forall \ell\in [t_j].
\end{equation}
Note that (\ref{eq:WWEF-utils}) holds trivially if $\tau_\ell=0$ since both sides are zero.

To prove the theorem, we prove the following auxiliary claim for all $s\in[t_j]$:
\begin{align}
\label{eq:claim}
y\cdot \max_{1\le r\le t_j}\beta_r &+ \sum_{\ell=1}^s
\alpha_{\ell} \ge \rho\left(\sum_{\ell=1}^s\beta_\ell - x\beta_1 \right) + \left(y+\sum_{\ell=1}^s\tau_\ell - \rho(s-x)\right)\max_{s\le r\le t_j}\beta_r.
\end{align}
To prove \eqref{eq:claim}, we use induction on $s$.
For the base case $s=1$, apply (\ref{eq:WWEF-utils}) with $\ell=1$, adding the term $y\cdot \max_{1\le r\le t_j}\beta_r$ to both sides:
\begin{align*}
y\cdot \max_{1\le r\le t_j}\beta_r + \alpha_1
&\ge (y+\tau_1)\cdot  \max_{1\le r\le t_j}\beta_r \\
&\ge \rho y\beta_1 + (y+\tau_1-\rho y)
\cdot \max_{1\le r\le t_j}\beta_r \\
&= \rho(1-x)\beta_1 + (y+\tau_1-\rho(1-x))\cdot \max_{1\le r\le t_j}\beta_r.
\end{align*}
For the inductive step, assume that \eqref{eq:claim} holds for some $s-1\ge 1$; we will prove it for $s$.
Using the inductive hypothesis for the first inequality below, we have
\begin{align*}
y\cdot &\max_{1\le r\le t_j}\beta_r + \sum_{\ell=1}^s \alpha_{\ell} \\
&= y\cdot \max_{1\le r\le t_j}\beta_r + \sum_{\ell=1}^{s-1}\alpha_{\ell} + \alpha_{s} \\
&\ge \rho\left(\sum_{\ell=1}^{s-1}\beta_\ell - x\beta_1\right) + \left(y+\sum_{\ell=1}^{s-1}\tau_\ell - \rho(s-1-x)\right) \cdot\max_{s-1\le r\le t_j}\beta_r + \alpha_{s}
\\
&\overset{\eqref{eq:WWEF-utils}}{\ge} \rho\left(\sum_{\ell=1}^{s-1}\beta_\ell - x\beta_1\right) + \left(y+\sum_{\ell=1}^{s-1}\tau_\ell - \rho(s-1-x)\right) \cdot\max_{s-1\le r\le t_j}\beta_r  + \tau_s\max_{s \le r \le t_j}\beta_r 
\\
&\ge \rho\left(\sum_{\ell=1}^{s-1}\beta_\ell - x\beta_1\right) + \left(y+\sum_{\ell=1}^{s-1}\tau_\ell - \rho(s-1-x)\right)  \cdot\max_{s\le r\le t_j}\beta_r  + \tau_s\max_{s\le r\le t_j}\beta_r 
&&\text{(*)}
\\
&= \rho\left(\sum_{\ell=1}^{s-1}\beta_\ell - x\beta_1\right) + \left(y+\sum_{\ell=1}^s\tau_\ell - \rho(s-1-x)\right) \cdot\max_{s\le r\le t_j}\beta_r \\
&= \rho\left(\sum_{\ell=1}^{s-1}\beta_\ell - x\beta_1\right) + \rho\max_{s\le r\le t_j}\beta_r   + \left(y+\sum_{\ell=1}^s\tau_\ell - \rho(s-x)\right)
\cdot \max_{s\le r\le t_j}\beta_r \\
&\ge \rho\left(\sum_{\ell=1}^{s-1}\beta_\ell - x\beta_1\right) + \rho\beta_s  + \left(y+\sum_{\ell=1}^s\tau_\ell - \rho(s-x)\right)\cdot \max_{s\le r\le t_j}\beta_r \\
&= \rho\left(\sum_{\ell=1}^{s}\beta_\ell - x\beta_1\right) + \left(y+\sum_{\ell=1}^s\tau_\ell - \rho(s-x)\right) \cdot \max_{s\le r\le t_j}\beta_r.
\end{align*}
The inequality denoted by (*) follows from the fact that, by \eqref{eq:WWEF-picks},
the expression in the rightmost parentheses is at least $0$, and  $\max_{s-1\le r\le t_j}\beta_r \ge \max_{s\le r\le t_j}\beta_r$.
This completes the induction and establishes \eqref{eq:claim}.

Using \eqref{eq:claim} with $s = t_j$ and \eqref{eq:WWEF-picks}, we get that
\begin{align*}
y\cdot \max_{1\le r\le t_j}\beta_r + \sum_{\ell=1}^{t_j}
\alpha_{\ell}
~\ge~ \rho\left(\sum_{\ell=1}^{t_j}\beta_\ell - x\beta_1 \right).
\end{align*}

Letting $A_i$ and $A_j$ be the bundles of agents $i$ and $j$ after agent $j$'s $t_j$-th pick, and $g$ be the item in $A_j$ for which agent $i$ has the highest utility (so that $u_i(g) = \max_{1\le r\le t_j}\beta_r$), we have 
\begin{align*}
y\cdot u_i(g) + u_i(A_i) 
~\ge~
\frac{w_i}{w_j}\cdot (u_i(A_j) - x\beta_1) 
~\ge~
\frac{w_i}{w_j}\cdot (u_i(A_j) - x\cdot u_i(g)).
\end{align*}
Hence the \wefxy condition for agent $i$ towards agent $j$ is fulfilled, completing the proof.

\subsection*{Proof of \Cref{thm:picking-WPROP}}

Our proof proceeds in a similar way as the WPROP1 proof in Theorem~3.3 of \citet{ChakrabortyScSu21}.
For notational convenience, let $y := 1-x$.

($\Rightarrow$) 
Assume that $\pi$ fulfills WPROP$(x,y)$.
In particular, for every integer $k\geq 1$, WPROP$(x,y)$ holds for an instance where every agent has utility $1$ for some $k$ items and $0$ for the remaining $m-k$ items.
The $k$ valuable items are picked in the first $k$ turns, so if we define $t_i$ based on this prefix, agent $i$ has utility $t_i$ for bundle $A_i$ and $k$ for the entire set of items.
If $t_i=k$ then the claim trivially holds.
Otherwise, the most valuable item $g$ not in $A_i$ has a value of $1$, 
so by definition of WPROP$(x,y)$,
\begin{align*}
t_i+y
=
u_i(A_i) + y\cdot u_i(g)
\geq
\frac{w_i}{w_N}\cdot (u_i(M) - n\cdot x\cdot u_i(g))
=
\frac{w_i}{w_N}\cdot (k-nx).
\end{align*}
It follows that $t_i + (1-x)\geq \frac{w_i}{w_N}\cdot (k-nx)$.

($\Leftarrow$) 
Let $\rho := w_i/w_N$, and assume that in $\pi$, agent $i$ picks $t_i$ times.
Denote by $\tau_0$ the number of times other agents pick before agent~$i$'s first pick, and $\tau_\ell$ the number of times other agents pick after agent $i$'s $\ell$-th pick and before her $(\ell+1)$st pick for $1\leq\ell\leq t_i$, where we insert a dummy empty pick for agent $i$ at the end of the picking sequence.
Let agent $i$'s utility for the items that other agents pick in phase $\ell$ be $\beta^\ell_1,\beta^\ell_2,\dots,\beta^\ell_{\tau_\ell}$, where the phases $0,1,\dots,t_i$ are defined as in the previous sentence.
Moreover, let agent $i$'s utility for the items that she picks be $\alpha_1,\dots,\alpha_{t_i}$, respectively.
We may assume that at least one other agent has at least one pick; otherwise the WPROP$(x,y)$ condition holds trivially.
Let $B$ be a singleton set consisting of an item in $M\setminus A_i$ for which agent~$i$ has the highest utility.
For the prefix before agent $i$'s first pick, the condition in the theorem statement implies that $y\ge \rho(\tau_0-nx)$, i.e., $y+\rho nx \ge \rho\tau_0$.
By definition of $B$, we have $u_i(B)\ge \beta^\ell_j$ for all $0\leq\ell\leq t_i$ and $1\leq j\leq \tau_\ell$, so $\tau_0u_i(B)\ge \sum_{j=1}^{\tau_0}\beta^0_j$.
It follows that 
\begin{align}
(y+\rho nx)\cdot u_i(B) 
&= (\rho\tau_0 + (y+\rho nx-\rho\tau_0))\cdot u_i(B) \nonumber \\
&\ge \rho\sum_{j=1}^{\tau_0}\beta^0_j + (y+\rho nx-\rho\tau_0)\max_{\substack{0\le \ell\le t_i \\ 1\le j\le \tau_\ell}}\beta^\ell_j.
\label{eq:WPROP-init}
\end{align}

We claim that for each $0\leq s\leq t_i$,
\begin{align*}
(y+\rho nx)\cdot u_i(B) &+ (1-\rho)\sum_{\ell=1}^s \alpha_\ell \\ 
&\ge \rho\sum_{\ell=0}^s\sum_{j=1}^{\tau_\ell}\beta^\ell_j + \left(y+\rho nx+s(1-\rho)-\rho\sum_{\ell=0}^s \tau_\ell\right)\max_{\substack{s\le \ell\le t_i \\ 1\le j\le \tau_\ell}}\beta^\ell_j.
\end{align*}
To prove the claim, we use induction on $s$.
The base case $s=0$ follows immediately from (\ref{eq:WPROP-init}) since the terms $(1-\rho)\sum_{\ell=1}^s \alpha_\ell$ and $s(1-\rho)$ vanish.
For the inductive step, assume that the claim holds for some $s-1\ge 0$; we will prove it for $s$.
For the prefix before agent $i$'s $s$-th pick, the condition in the theorem statement implies that $(s-1) + y \ge \rho\left(s-1 + \sum_{\ell=0}^{s-1}\tau_\ell - nx\right)$, or equivalently
\begin{equation}
\label{eq:WPROP-picks}
(y + \rho nx)+(s-1)(1-\rho)\ge \rho\sum_{\ell=0}^{s-1} \tau_\ell.
\end{equation}
Every time agent $i$ picks, she picks an item with the highest utility for her among the available items, which means that
\begin{equation}
\label{eq:WPROP-utils}
\alpha_s \geq \beta^s_j
\qquad\qquad \forall j\in[\tau_s].
\end{equation}
We have
\begin{align*}
&(y + \rho nx)\cdot u_i(B) + (1-\rho)\sum_{\ell=1}^s \alpha_\ell \\
&= (y + \rho nx)\cdot u_i(B) + (1-\rho)\sum_{\ell=1}^{s-1}\alpha_\ell + (1-\rho)\alpha_s \\
&\ge \rho\sum_{\ell=0}^{s-1}\sum_{j=1}^{\tau_\ell}\beta^\ell_j + \left(y + \rho nx+(s-1)(1-\rho)-\rho\sum_{\ell=0}^{s-1} \tau_\ell\right)  \cdot \max_{\substack{s-1\le \ell\le t_i \\ 1\le j\le \tau_\ell}} \beta^\ell_j + (1-\rho)\alpha_s \\
&\ge \rho\sum_{\ell=0}^{s-1}\sum_{j=1}^{\tau_\ell}\beta^\ell_j + \left(y + \rho nx+(s-1)(1-\rho)-\rho\sum_{\ell=0}^{s-1} \tau_\ell\right)  \cdot \max_{\substack{s\le \ell\le t_i \\ 1\le j\le \tau_\ell}}\beta^\ell_j + (1-\rho)\alpha_s \\
&= \rho\sum_{\ell=0}^{s-1}\sum_{j=1}^{\tau_\ell}\beta^\ell_j +  \left(y + \rho nx+s(1-\rho)-\rho\sum_{\ell=0}^s \tau_\ell\right)  \cdot \max_{\substack{s\le \ell\le t_i \\ 1\le j\le \tau_\ell}}\beta^\ell_j  + \left(\rho\tau_s - (1-\rho)\right)\max_{\substack{s\le \ell\le t_i \\ 1\le j\le \tau_\ell}}\beta^\ell_j \\
&\quad + (1-\rho)\alpha_s \\
&\ge \rho\sum_{\ell=0}^{s-1}\sum_{j=1}^{\tau_\ell}\beta^\ell_j +  \left(y + \rho nx+s(1-\rho)-\rho\sum_{\ell=0}^s \tau_\ell\right)  \cdot \max_{\substack{s\le \ell\le t_i \\ 1\le j\le \tau_\ell}}\beta^\ell_j  + \frac{\rho\tau_s - (1-\rho)}{\tau_s}\cdot\sum_{j=1}^{\tau_s}\beta_j^s \\
&\quad + \frac{1-\rho}{\tau_s}\cdot\sum_{j=1}^{\tau_s}\beta_j^s \\
&= \rho\sum_{\ell=0}^{s-1}\sum_{j=1}^{\tau_\ell}\beta^\ell_j +  \left(y + \rho nx+s(1-\rho)-\rho\sum_{\ell=0}^s \tau_\ell\right)  \cdot \max_{\substack{s\le \ell\le t_i \\ 1\le j\le \tau_\ell}}\beta^\ell_j  + \rho\sum_{j=1}^{\tau_s}\beta_j^s \\
&= \rho\sum_{\ell=0}^{s}\sum_{j=1}^{\tau_\ell}\beta^\ell_j + \left(y + \rho nx+s(1-\rho)-\rho\sum_{\ell=0}^s \tau_\ell\right)  \cdot \max_{\substack{s\le \ell\le t_i \\ 1\le j\le \tau_\ell}}\beta^\ell_j.
\end{align*}
Here, the first inequality follows from the inductive hypothesis, the second from (\ref{eq:WPROP-picks}) and the fact that $\max_{\substack{s-1\le \ell\le t_i \\ 1\le j\le \tau_\ell}}\beta^\ell_j \ge \max_{\substack{s\le \ell\le t_i \\ 1\le j\le \tau_\ell}}\beta^\ell_j$, and the third from (\ref{eq:WPROP-utils}).
This completes the induction and establishes the claim.

Finally, taking $s=t_i$ in the claim, we get
\begin{align*}
(y + \rho nx)\cdot u_i(B) &+ (1-\rho)\sum_{\ell=1}^{t_i} \alpha_\ell \\
&\ge \rho\sum_{\ell=0}^{t_i}\sum_{j=1}^{\tau_\ell}\beta^\ell_j + \left(y + \rho nx+t_i(1-\rho)-\rho\sum_{\ell=0}^{t_i} \tau_\ell\right) \cdot \max_{1\le j\le \tau_{t_i}}\beta^\ell_j \\
&\ge \rho\sum_{\ell=0}^{t_i}\sum_{j=1}^{\tau_\ell}\beta^\ell_j,
\end{align*}
where the second inequality follows from the condition in the theorem statement for the entire picking sequence.
Adding $\rho\sum_{\ell=1}^{t_i} \alpha_\ell$ to both sides, we find that
\[
(y + \rho nx)\cdot u_i(B) + \sum_{\ell=1}^{t_i} \alpha_\ell \ge \rho\left(\sum_{\ell=1}^{t_i} \alpha_\ell + \sum_{\ell=0}^{t_i}\sum_{j=1}^{\tau_\ell}\beta^\ell_j\right).
\]
In other words, $(y + \rho nx)\cdot u_i(B) + u_i(A_i) \ge \rho\cdot u_i(M)$.
Rearranging this yields $u_i(A_i) + y\cdot u_i(B) \ge \rho\cdot (u_i(M) - n\cdot x\cdot u_i(B))$.
This means that the allocation produced by $\pi$ is WPROP$(x,y)$, as desired.

\section{Weighted Egalitarian Rule with Binary Valuations}
\label{app:WEG-binary}
In this appendix, we prove that the WEG rule has desirable properties in the setting of \emph{binary (additive)} utilities, which is slightly more general than that of identical items. 
In the binary setting, 
for each agent~$i$ there is a set $Z_i\subseteq M$
of items with a value of $1$, and the value of the other items is~$0$. Denote $z_i := |Z_i| = $ the number of items valued positively by agent~$i$.
We assume without loss of generality that each item is valued positively (at $1$) by at least one agent, and that each agent positively values at least one item. 
For convenience, we assume throughout this appendix that the weights are normalized so that $w_N = 1$.

The proportional share of agent $i$ is $q_i := \frac{w_i}{w_N}\cdot z_i = w_i z_i$.
Since the OMMS and APS are at most the proportional share, and both of them are integers, we have that 
both shares are at most 
$\floor{q_i}$.
We will prove that every WEG allocation guarantees to each agent $i$ a value of at least $\floor{q_i}$, which means that it is APS-fair and OMMS-fair. To this end, we need some lemmas.

\begin{lemma}
\label{lem:quotagraph}
In each allocation $A$,
if $|A_i\cap Z_i| < w_i z_i$  for some agent $i$,
then $|A_j \cap Z_i| > w_j z_i$ for some other agent $j$.
\end{lemma}
\begin{proof}
We can decompose $|Z_i|$ into a sum of $n$ elements in two ways:
\begin{align*}
|Z_i| 
&= 
\sum_{j=1}^n |A_j\cap Z_i| && \text{since the $A_j$'s are pairwise disjoint and their union is $M$;}
\\
|Z_i|
&= 
z_i =
\sum_{j=1}^n w_j z_i &&\text {since $\sum_{j=1}^n w_j = 1$.}
\end{align*}
Subtracting the term $|A_i\cap Z_i|$ from both sums gives:
\begin{align*}
\sum_{j \neq i} |A_j\cap Z_i| 
&=
\sum_{j \neq i} w_j z_i
+ 
w_i z_i - |A_i\cap Z_i|
\\
&>
\sum_{j \neq i} w_j z_i && \text{since $|A_i\cap Z_i|<w_i z_i$,}
\end{align*}
so there is some $j\neq i$ such that
$
|A_j\cap Z_i| > w_j z_i,
$
as claimed.
\end{proof}

We define the \emph{quota graph} of an allocation $A$ as the directed graph in which the vertices correspond to the agents, and there is an edge from $i$ to $j$ if and only if $|A_i\cap Z_i| < w_i z_i$ and $|A_j \cap Z_i| > w_j z_i$.
Lemma~\ref{lem:quotagraph} implies that each agent $i$ with $|A_i \cap Z_i| < w_i z_i$ has at least one outgoing edge. 

In what follows, we consider only \emph{wasteless} allocations, that is, allocations in which each item is allocated to an agent who values it at $1$. In other words, $A_i\subseteq Z_i$ for all $i\in N$.
Clearly, every WEG allocation is wasteless.

\begin{lemma}
\label{lem:quotagraph-path}
Let $A$ be any wasteless allocation. 
Suppose that in the quota graph of $A$, there 
are edges $i\to j\to k$ (where possibly $i = k$).
Then $z_i < z_j$.
\end{lemma}
\begin{proof}
The edge entering $j$ implies
\begin{align*}
w_j z_i < |A_j\cap Z_i|
&\leq
|A_j| = |A_j\cap Z_j|,
\end{align*}
while the edge exiting $j$ implies
\begin{align*}
|A_j \cap Z_j| &< w_j z_j.
\end{align*}
Combining the two inequalities gives $z_i < z_j$. 
\end{proof}
Since the inequality in Lemma~\ref{lem:quotagraph-path} cannot hold simultaneously for all pairs of adjacent vertices in a directed cycle, we have the following corollary.
\begin{corollary}
\label{lem:quotagraph-acyclic}
The quota graph of a wasteless allocation has no directed cycles.
\end{corollary}

Let $a_i := |A_i| = |A_i\cap Z_i|$, and let 
$d_i := \frac{u_i(A_i)}{u_i(M)} - \frac{w_i}{w_N} = \frac{a_i}{z_i} - w_i$ for all $i\in N$.
By definition, a WEG allocation maximizes the leximin vector of the $d_i$'s. 

\begin{theorem}
With binary valuations, every WEG allocation gives each agent $i$ a value of at least $\floor{q_i} = \floor{w_i z_i}$.
\end{theorem}
\begin{proof}
Let $A$ be a WEG allocation (which is necessarily wasteless). Suppose for contradiction that for some agent, say agent $1$,
we have $a_1 < \floor{q_1}$.
Since both quantities are integers, 
this implies 
$a_1 \leq \floor{q_1} - 1 = \floor{w_1z_1} - 1 \leq w_1 z_1 - 1$, and so $d_1 \leq - 1 / z_1$.

We show that it is possible to ``transfer'' one unit of utility from some agent $k$ to agent~$1$, so that the leximin vector of the $d_i$'s strictly improves---this contradicts the assumption that the original allocation $A$ is a WEG allocation.
More precisely, we show that, for some agent $k\neq 1$,
there is a different allocation~$A'$ in which
all agents except agents~$1$ and $k$ have the same utility as in $A$,
$a_k' = a_k-1$,
$a_1' = a_1 + 1$,
and $d_k' > - 1 / z_1$.
This implies that the leximin vector of $A'$ is better than the leximin vector of $A$.

Starting from agent $1$, we will renumber the remaining agents.
In particular, we follow the edges of the quota graph (picking an outgoing edge arbitrarily if there are two or more), and number the agents $2$, $3$, and so on.
\Cref{lem:quotagraph-acyclic} implies that
this walk must end at some agent $k$ with no outgoing edges, that is, $a_k \geq w_k z_k$.
Note that, by definition of the quota graph,
$|A_{j+1}\cap Z_j| > w_{j+1} z_j$ for all $j\in[k-1]$; in particular, 
$A_{j+1}\cap Z_j\neq\emptyset$.

For each $j\in [k-1]$, we transfer one item 
from the set 
$A_{j+1}\cap Z_j$ 
(which is currently allocated to agent $j+1$) to agent $j$.
This sequence of transfers 
decreases the utility of agent~$k$ by $1$,
increases the utility of agent~$1$ by $1$, 
and does not change the utility of any other agent.
We derive the required lower bound on $d_k'$ in two cases.

\underline{Case 1}: $z_k\leq z_1$.
After the transfer, we have
\begin{align*}
d_k' &= \frac{|A_k|-1}{z_k} - w_k
\\
&\geq \frac{|A_k\cap Z_{k-1}|-1}{z_k} - w_k
\\
&> \frac{w_k z_{k-1} - 1}{z_k} - w_k &&
\text{since there is an edge $(k-1) \to k$.}
\end{align*}
From Lemma~\ref{lem:quotagraph-path}, we have
$z_1 < \cdots < z_{k-1}$, so the above implies that
\begin{align*}
d_k' &> 
\frac{w_k z_1 - 1}{z_k} - w_k
\\
& = 
w_k z_1 \left(\frac{1}{z_k} - \frac{1}{z_1}\right)
- \frac{1}{z_k}
\\
& = 
w_k z_1 \left(\frac{1}{z_k} - \frac{1}{z_1}\right)
- \frac{1}{z_1}
- \left(\frac{1}{z_k} - \frac{1}{z_1}\right)
\\
& = 
\left(w_k z_1 - 1\right)\left(\frac{1}{z_k} - \frac{1}{z_1}\right)
- \frac{1}{z_1}.
\end{align*}
Since
$z_k\leq z_1$ by the assumption of Case~1, it holds that
 $1/z_k - 1/z_1\geq 0$.
If $w_k z_1 -1 \geq 0$ too, then the above inequality implies $d_k' > -1/z_1$.
Else, $w_k z_1 - 1 < 0$, and regardless of the above inequality we have 
$d_k' = |A_k'| / z_k - w_k \geq 0 - w_k > -1/z_1$.

\underline{Case 2}: $z_k >z_1$.
After the transfer, we have
\begin{align*}
d_k' &= \frac{a_k-1}{z_k} - w_k
\\
&=
\frac{a_k - w_k z_k - 1}{z_k}
\\
&\geq 
\frac{0 - 1}{z_k} && \text{$a_k\geq w_k z_k$, since there are no outgoing edges from $k$}
\\
&>
-1 / z_1 && \text{since $z_k > z_1$ by the assumption of Case~2.}
\end{align*}

In both cases, we have $d_k' > -1/z_1$, and so the transfer strictly improves the leximin vector.
This yields the desired contradiction.
\end{proof}

\begin{corollary}
\label{cor:weg_aps}
With binary valuations, every WEG allocation is APS-fair and OMMS-fair.
\end{corollary}

It may be interesting to extend these results from binary additive valuations to \emph{submodular valuations with binary marginals} (also known as \emph{matroid rank functions}); these valuations have been studied, e.g., by \citet{babaioff2021fair}, \citet{barman2021existence}, \citet{benabbou2021finding}, and \citet{MontanariScSu24}.

\end{document}